\documentclass{article} 
\usepackage[hmargin=3.5cm,vmargin=4cm]{geometry}

\newcommand{\bbeta}{{\bm \beta}}
\newcommand{\bX}{{\bf X}}
\newcommand{\by}{{\bf y}}
\newcommand{\R}{\mathbb{R}}
\newcommand{\V}[2]{\mathbb{V}_{#1} \left[ #2 \right] }
\newcommand{\E}[2]{\mathbb{E}_{#1} \left[ #2 \right] }

\newcommand{\btheta}{{\bm \theta}}
\newcommand{\bEta}{{\bm \eta}}

\newcommand{\bT}{\textbf{T}}

\newcommand{\ba}{\textbf{a}}

\newcommand{\trans}{^{\rm T}}
\newcommand{\inv}{^{-1}}

\RequirePackage{amsthm,amsmath,amsfonts,amssymb}
\RequirePackage[numbers]{natbib}
\RequirePackage[colorlinks,citecolor=blue,urlcolor=blue,backref=page,backref=page]{hyperref}
\RequirePackage{graphicx}

\usepackage{enumitem}
\usepackage{changepage}
\usepackage{nicefrac}
\usepackage{bm}
\usepackage{adjustbox}
\usepackage{subcaption}
\usepackage{float}
\usepackage{threeparttable}
\usepackage{enumitem}
\usepackage{makecell}
\usepackage{boldline}
\usepackage{multirow}
\usepackage{float}
\usepackage{booktabs}
\usepackage{caption}
\usepackage{pdflscape}
\usepackage{array}
\usepackage[ruled, lined, linesnumbered, longend]{algorithm2e}
\usepackage{multicol}
\usepackage{ltablex}

\theoremstyle{plain}

\newtheorem{theorem}{Theorem}[section]

\newtheorem{proposition}[theorem]{Proposition}

\newcolumntype{C}[1]{>{\centering\let\newline\\\arraybackslash\hspace{0pt}}m{#1}}

\title{GRASP: Grouped Regression with Adaptive Shrinkage Priors}

\author{Shu Yu Tew $^1$ \\ \href{mailto:shu-yu.tew@monash.edu}{shu-yu.tew@monash.edu} 
   \and Daniel F. Schmidt $^1$ \\ \href{mailto:daniel.schmidt@monash.edu}{daniel.schmidt@monash.edu} 
   \and Mario Boley $^{1, 2}$ \\ \href{mailto:mario.boley@monash.edu}{mario.boley@monash.edu}\\ \href{mailto:mboley@is.haifa.ac.il}{mboley@is.haifa.ac.il}}

\date{%
    $^1$Department of Data Science and AI, Monash University\\%
    $^2$Department for Information Systems, University of Haifa\\[2ex]%
}

\begin{document}

\maketitle

\begin{abstract}
We introduce GRASP, a simple Bayesian framework for regression with grouped predictors, built on the normal beta prime (NBP) prior. The NBP prior is an adaptive generalization of the horseshoe prior with tunable hyperparameters that control tail behavior, enabling a flexible range of sparsity, from strong shrinkage to ridge-like regularization. Unlike prior work that introduced the group inverse-gamma gamma (GIGG) prior by decomposing the NBP prior into structured hierarchies, we show that directly controlling the tails is sufficient without requiring complex hierarchical constructions. Extending the non-tail adaptive grouped half-Cauchy hierarchy of \citet{xu2016bayesian}, GRASP assigns the NBP prior to both local and group shrinkage parameters allowing adaptive sparsity within and across groups. A key contribution of this work is a novel framework to explicitly quantify correlations among shrinkage parameters within a group, providing deeper insights into grouped shrinkage behavior. We also introduce an efficient Metropolis-Hastings sampler for hyperparameter estimation. Empirical results on simulated and real-world data demonstrate the robustness and versatility of GRASP across grouped regression problems with varying sparsity and signal-to-noise ratios.
\end{abstract}

\section{Introduction}
\label{sec:intro}

Group structures are common in regression analysis. They can appear in the form of categorical predictors represented by groups of dummy variables or in the context of additive modeling, where each predictor can be expressed as a set of basis functions forming a group; in applications such as gene expression analysis and financial market modeling, groupings exist naturally in the data. For instance, genes that influence similar traits form groups in gene expression data, while stocks from the same sector form groups in financial data. In these scenarios, group shrinkage plays an important role: when there is insufficient evidence to suggest the significance of predictors within a group, the entire group of predictors is shrunk towards zero. This reduces the noise from individual ``spurious predictors'', which tend to appear more frequently in high-dimensional settings, and decreases model complexity, thereby reducing the risk of overfitting.

Within the Bayesian framework, there has been extensive research focusing on the application of continuous shrinkage priors for linear regression problems involving group predictor variables. Traditional approaches, such as the group lasso\cite{yuan2006model, simon2013sparse}, the group bridge \cite{mallick2017bayesian}, and the group horseshoe \cite{xu2016bayesian} primarily apply shrinkage at the group level and do not consider within-group shrinkage. Given a data matrix $\bX \in \R^{n \times p}$ with $G$ groups of predictors, \citet{xu2016bayesian} introduced a hierarchical Bayesian grouped model for the corresponding target variable $\by \in \R^n$ that integrates group-level local shrinkage parameters along with the global and individual local shrinkage parameters:
\begin{align}\label{eq: Bayesian_group_regression_hierarchy}
\begin{split}
  \by \; | \; \bX, \bm{\beta}, \sigma^2 \; &\sim \; N_n\left(\bX \bbeta, \; \sigma^2\bm{I}_n\right) \\
  \sigma^2 \; &\sim \; \sigma^{-2} d\sigma^2
\end{split}\\
\begin{split} \label{eq: group_shrinkage_hier}
  \beta_{jg}|\tau^2, \lambda_{jg}^2, \delta^2_g, \sigma^2 \; &\sim \; N\left(0, \; \tau^2 \lambda^2_{jg} \delta^2_g \sigma^2\right) \\ 
    \lambda^2_j \; &\sim \; \pi(\lambda^2_j)d\lambda^2_j \\
    \delta^2_g \; &\sim \; \pi(\delta^2_g)d\delta^2_g \\
    \tau \; &\sim \; \pi(\tau)d\tau.
\end{split}
\end{align}
where $\bbeta = (\beta_1, \cdots, \beta_G) \in \R^G$ are the regression coefficients, $\delta_g$ is the local shrinkage parameter for the $g^{\rm th}$ group, controlling the shrinkage at group level, $\lambda$ represents the local shrinkage parameter that controls individual shrinkage, while $\tau$ is the global shrinkage parameter that controls the overall shrinkage. This hierarchical structure allows for simultaneous control of shrinkage between and within groups. In their approach, a half-Cauchy prior is used to independently model the local shrinkage parameter at both the individual and group levels. For a group size of one (i.e. $G=p$), this prior configuration simplifies to the horseshoe+ prior; we therefore refer to it as the group horseshoe+ prior. \citet{boss2023group} further built on this framework by proposing the group inverse-gamma gamma (GIGG) prior. This work utilizes a similar hierarchical structure but assigns inverse-gamma and gamma distributions to the individual and group local shrinkage parameters, respectively. This setup results in a normal beta prime (NBP) shrinkage prior - a generalization of the group horseshoe prior - on the regression coefficient marginally, introducing hyperparameters that control the sparsity level and allowing adaptive shrinkage to be learned from the data. However, this hierarchy has several limitations: (i) the decomposition of the beta prime prior restricts its applicability to overlapping group structures; and (ii) the proposed Gibbs sampler, which involves sampling from the generalized inverse Gaussian (GIG) distribution, can be computationally inefficient.

In this paper, we present GRASP - Grouped Regression with Adaptive Shrinkage Priors - an efficient alternative for group regression using the NBP prior within the hierarchical Bayesian grouped model proposed by \citet{xu2016bayesian}. We assign a beta prime prior for both the individual and group local shrinkage parameters. Unlike the method by \citet{boss2023group}, which ensures the product of the local shrinkage priors is a beta prime, we demonstrate that this condition is not necessary. It is not crucial for the marginal shrinkage prior on the beta coefficient to take a specific form, such as the normal beta prime.  Instead, by simply assigning the beta prime prior to both the local shrinkage parameters, we can achieve similar results. We emphasize that controlling the tail behavior of the priors by adjusting the hyperparameters is more important than restricting the exact marginal form of the shrinkage prior. This concept enables the straightforward extension of adaptive group shrinkage regression to multiple and overlapping groups for the first time (contrary to what was previously mentioned as a limitation in \cite{boss2023group}). Additionally, we examine the log-scale interpretation of the local shrinkage parameters under the beta prime prior. This logarithmic transformation offers a clear understanding of the interaction between individual and group-level shrinkage, serving as a guide for selecting appropriate shrinkage priors. It reveals how different prior choices affect the correlation between shrinkage parameters and how the hyperparameters influence this relationship. 

\noindent In summary, the main contributions of this manuscript are:
\begin{itemize}
    \item We introduced GRASP, a Bayesian estimator for the normal beta prime (NBP) model in the linear regression contexts that handles both within and between grouping structures. Notably, the proposed framework introduces adaptivity in group shrinkage regression such that straightforward extension to multiple and overlapping group structures is allowed, although examples of overlapping groups were not explicitly demonstrated in this paper.

    \item Developed an efficient Metropolis-Hastings algorithm for estimating the hyperparameters of the NBP prior within the Gibbs sampling framework. The method builds on the general strategy of approximating nonstandard Bayesian posterior distributions using exponential family proposals, as introduced by \citet{salimans2013fixed}, and enables accurate and scalable sampling of shape parameters.

    
    \item We show that our method outperforms existing adaptive beta prime estimators in terms of computational speed and predictive performance, across linear regression problems with and without group structures. 
    
    \item Investigation of the log-scale interpretation of local shrinkage parameters under the beta prime prior, providing insights into the interaction between individual and group-level shrinkage.
\end{itemize}
The remainder of this paper is organized as follows: Section~\ref{sec:GeneralizedHorseshoePrior} introduces the NBP prior and its established theoretical properties. We then present the log-scale interpretation of the shrinkage parameters (Sec.~\ref{sec:LogScaleInterpretation}) and compare the NBP prior to other well-known priors, specifically the Laplace and Student-t priors (Sec.~\ref{sec:Relation to other}) in this framework. Section~\ref{sec:Grouped GHS model} details the grouped NBP model and Section~\ref{sec:posterior sampling} outlines the sampling algorithm required for posterior estimation, including the
conditional posterior distributions needed for parameter estimation. Section~\ref{sec:Results} presents an empirical comparison of the proposed and existing NBP estimators for linear regression problems with and without grouped predictors. Finally, Section~\ref{sec:realWorldDataset} evaluates the proposed method on real-world datasets, comparing its performance against existing approaches with discussions on its practical applicability and limitations.

\section{Normal Beta Prime Prior} \label{sec:GeneralizedHorseshoePrior}

The normal beta prime (NBP) prior \cite{bai2021beta} has been referred to by various names in the literature, such as the three-parameter beta (TPB) prior \cite{armagan2011generalized}, the adaptive normal hypergeometric inverted beta prior \cite{yu2021adaptive}, the inverse-gamma gamma prior \cite{bai2017inverse, boss2023group}, and the generalized horseshoe prior \cite{schmidt2019bayesian}. Despite these different names, they all describe the same underlying prior model, which involves placing a beta prime distribution over the local shrinkage parameter, $\lambda$:
\begin{equation} \label{eq: lambda_beta_prime}
   \pi(\lambda_j|a, b) = \frac{2 \lambda_j^{2a-1} (1+\lambda_j^2)^{-a-b}}{{\rm B}(a, b)}, \; a>0, b>0.
\end{equation}
where ${\rm B}(a, b)$ denotes the beta function. The hyperparameter $a$ controls the sparsity level of $\bm \beta$ (i.e. smaller $a$ values yield a marginal prior for $\bm \beta$ that concentrates more around ${\bm \beta} = 0$); while the hyperparameter $b$ can be seen as the tail-decay parameter (i.e. smaller values indicate a slower rate of decay at the tails of the marginal distribution)~\citep{schmidt2019bayesian}. By varying $a$ and $b$, many different prior distributions can be achieved, such as the Strawderman–Berger prior ($a = 0.5, b = 1$), the horseshoe prior ($a = 0.5, b = 0.5$), and normal-exponential-gamma (NEG) prior ($a = 1, b > 0$) \cite{armagan2011generalized}. 

\subsection{Theoretical Properties}
The theoretical studies discussed in this section primarily build upon the normal means problem whereby $ {\rm y}_j|  \beta_j \sim N(\beta_j,1)$. Here ${\by}$ is a $p$-dimensional mean vector and $\beta_j$ is the estimated mean associated with observation ${\rm y}_j$. This is a special case of the linear regression model described in hierarchy \eqref{eq: Bayesian_group_regression_hierarchy}, with $\bf X = I $, $p = n$, and $\sigma^2 = 1$. The normal means model serves as an important test case for the theoretical understanding of many shrinkage methods \cite{carvalho2010horseshoe, bhattacharya2015dirichlet, van2014horseshoe, van2017adaptive, rovckova2018bayesian, armagan2013generalized, bhadra2017horseshoe+, castillo2018empirical} due to its' simple and tractable framework. Within this framework, we examine and compare the shrinkage properties, focusing on sparsity and tail robustness, of the NBP prior against the well-known horseshoe prior.

\subsubsection{Efficiency for sparsity}

When the true parameter $\beta_0$ is zero, the horseshoe prior attains Kullback-Leibler super-efficiency \cite{carvalho2010horseshoe}. This implies that the estimated density using the horseshoe prior converges to the true density faster than the maximum likelihood estimator (MLE) under the assumption that the true mean vector is sparse. The horseshoe prior achieves this super-efficiency due to its asymptote at 0, which efficiently shrinks small coefficients towards zero. In the case of the NBP prior, \citet{bai2019large} demonstrated that the induced marginal prior $\pi(\beta_j | \sigma^2, \tau = 1, a, b)$ exhibits a pole at 0 if and only if $0 < a \leq 1/2$. The strength of this pole increases as $a \to 0$, indicating a stronger concentration of the prior around $\beta_j = 0$. And when $a > 1/2$, the pole disappears, resulting in a non-sparse prior. Consequently, \citet{yu2021adaptive} proved that the Bayes estimator under the NBP prior achieves super-efficiency at the origin with the Kullback-Leibler risk bound being smaller than that of the horseshoe prior when $0 < a < 1/2$.

\subsubsection{Tail robustness}
Another appealing property that the NBP prior shares with the horseshoe prior is tail robustness. \citet{boss2023group} established that for any combination of $a$ and $b$ values, the NBP prior has tails that decay at a polynomial rate, indicating heavy-tailed behavior. This heavy tail property of the prior contributes to its robustness in handling large signals or extreme observations.

Theorem 3.4 in \citet{yu2021adaptive} further supports the asymptotic tail robustness of the NBP estimator. The theorem demonstrates that as $|y| \to \infty$, the influence of the NBP prior on the estimate diminishes, resulting in an asymptotically unbiased estimator for large signals.

\subsection{Log-scale interpretation} \label{sec:LogScaleInterpretation}

A novel and insightful understanding of the properties of the NBP prior can be achieved by expressing the beta prime prior on the local shrinkage parameter in log-scale, denoted as $\xi = \log(\lambda)$. The probability density function for the transformed local-shrinkage variable, $\xi_j$ is then given by:
\begin{equation}\label{eq: xi_mu_nu}
   \pi(\xi_j|a, b) = \frac{2\sigma(2\xi_j)^{a} \sigma(-2\xi_j)^{b}}{{\rm B}(a, b)}.
\end{equation}
where $\sigma(x) = 1/(1+ e^{-x})$ is the standard logistic function. This density is known in the literature as the Type IV generalized logistic distribution. Table \ref{tab:priorDensity} lists the prior densities for $\bm \lambda$ and $\bm \xi$ implied by the different prior choices included in Figure \ref{fig: GHS_effects} and \ref{fig: LassoVSHS}. The concept of exploring the log-scale interpretation of shrinkage priors was first introduced by \citet{schmidt2018log}. It offers an intuitive understanding of the tail behavior and concentration properties of the prior distributions.

\begin{table*}[t]
\caption{Prior densities for $\bm \lambda$ and ${\bm \xi} = \log({\bm \lambda})$ associated with several popular shrinkage rules. Densities are given up to a constant.}
\label{tab:priorDensity}
\begin{center}
\begin{tabular}{lcc}
\toprule
Prior for $\bm \beta$ & Density for $\lambda$ & Density for ${\bm \xi}$\\
\midrule
Lasso & $2\lambda e^{-\lambda^2}$ & $2e^{-e^{2\xi}+2\xi}$ \\
Student-t & $\frac{b^a}{\Gamma(a)} 2 \lambda^{(-2a-1)} e^{\frac{b}{\lambda^2}}$ & $\frac{b^a}{\Gamma(a)} 2 e^{-2 \xi a - \frac{b}{e^{2\xi}}}$ \\
Horseshoe & $\frac{2}{\pi(1+\lambda^2)}$ & $\frac{2e^{\bm \xi}}{\pi (1+e^{2{\bm \xi}})}$ \\
Generalized Horseshoe & $\frac{2 \lambda_j^{2a-1} (1+\lambda_j^2)^{-a-b}}{{\rm B}(a, b)}$ &  $\frac{2\sigma(2\xi_j)^{a} \sigma(-2\xi_j)^{b}}{{\rm Beta}(a, b)}$ \\
\bottomrule
\end{tabular}   
\end{center}
\vspace{-4mm}
\end{table*}

The roles of the hyperparameters $a$ and $b$ for the NBP prior can be better reflected by reparameterizing them in terms of the mean ($\mu$) and sample size ($\nu$). Specifically, we can establish a relationship between $a$ and $b$ with $\mu$ and $\nu$ as $a = \mu \nu$ and $b = (1-\mu) \nu$ respectively. Figure \ref{fig: GHS_effects} depicts the beta prime prior in log scale with the $\mu$ and $\nu$ parameterization. Here, $\mu$ controls the skewness of the distribution: when $\mu < 0.5$, the distribution skews left (more probability placed on smaller $\lambda$), favoring sparsity and concentrating around ${\bm \beta} = 0$; Conversely, $\mu > 0.5$ results in a right-skewed distribution, allowing less sparsity and bias towards larger $\lambda$, thus avoiding over-shrinkage/underestimation of large effects; when $\mu = 0.5$, the distribution is symmetric, suggesting no preference over sparsity or non-sparsity. On the other hand, $\nu$ is the concentration parameter that controls the level of probability concentrated around $\mu$, with smaller values allowing for more variation. As $\nu \to \infty$, the prior converges to the ridge prior, with minimal variation in $\bm \lambda$ and a point mass at $\mu$. Consequently, by appropriately choosing $a$ and $b$, it becomes possible to approximate the behavior of many well-known priors across a wide range of sparsity levels, from those inducing sparsity to those promoting non-sparsity.

\begin{figure}[t]
     \centering
     \begin{subfigure}[b]{0.3\textwidth}
         \centering
         \includegraphics[width=\textwidth]{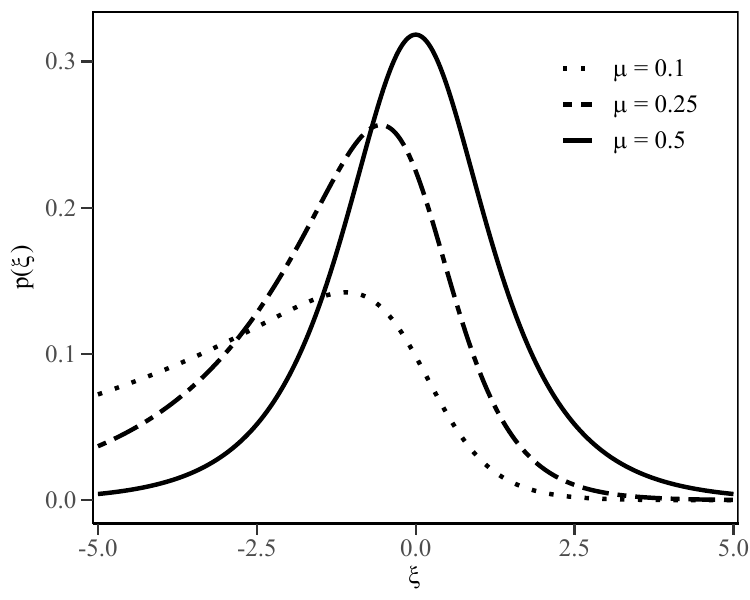}
         \caption{$\nu = 1$, $\mu \le 0.5$}
     \end{subfigure}
     \hspace{1em}%
     \begin{subfigure}[b]{0.3\textwidth}
         \centering
         \includegraphics[width=\textwidth]{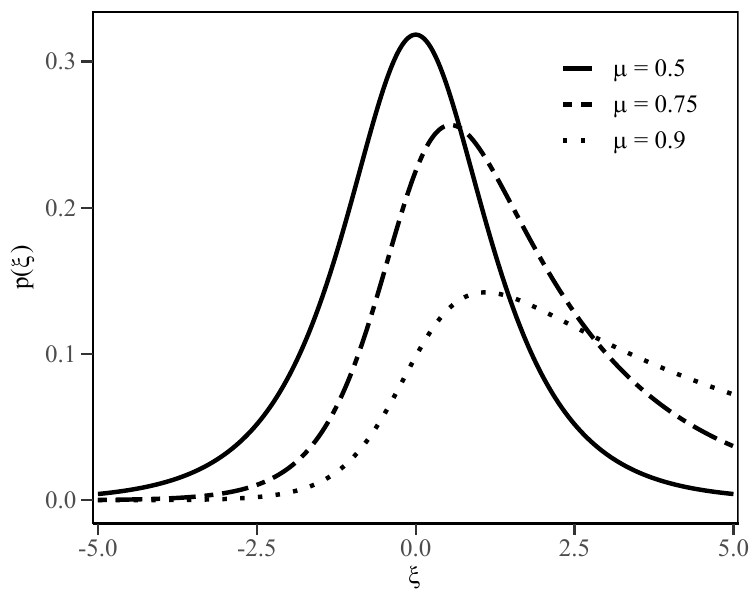}
         \caption{$\nu = 1$, $\mu \ge 0.5$}
     \end{subfigure}
     \hspace{1em}%
     \begin{subfigure}[b]{0.3\textwidth}
         \centering
         \includegraphics[width=\textwidth]{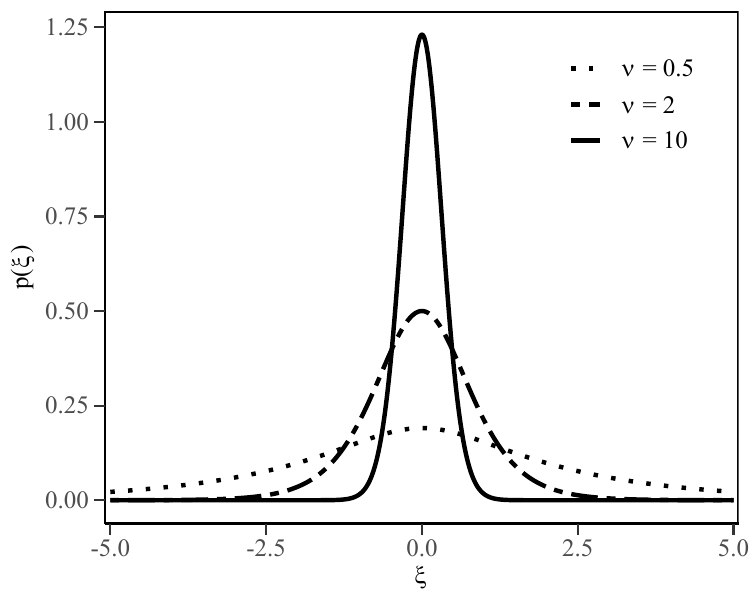}
         \caption{$\mu = 0.5$}
     \end{subfigure}
     \caption{Beta prime prior density in ${\bm \xi} = \log {\bm \lambda}$ space with varying $\mu$ and $\nu$.}
     \label{fig: GHS_effects}
\end{figure}

\subsection{Relation to other shrinkage priors} \label{sec:Relation to other}

Here we compare the characteristics and tail behavior of the NBP prior to popular shrinkage priors namely the Laplace prior and the Student-t prior using the log scale parameterization. The Laplace prior, which corresponds to the Bayesian Lasso model, is obtained by placing an exponential prior on $\lambda^2$ with a rate parameter of 1, i.e., $\lambda_j^2 \sim \rm{Exp}(1)$; While, the Student-t priors with a degree of freedom $\gamma$ can be achieved using inverse-gamma mixing, where $\lambda^2_j \sim {\rm IG}\left(\frac{\gamma}{2}, \frac{\gamma}{2}\right)$.

The log-scale parameterization provides a clearer visualization of why the Bayesian Lasso introduces bias in estimation when the underlying model coefficients are large, in contrast to the Horseshoe prior. In Figure \ref{fig: LassoVSHS(a)}, the marginal prior over $\beta$ for both the Horseshoe and Laplace prior appear symmetric, making the explanation less apparent. However, in Figure \ref{fig: LassoVSHS(b)}, we clearly observe that the distribution over $\xi$ for the Bayesian Lasso is asymmetric, with a heavier left-hand tail and a stronger preference for $\xi_j < 0$. This asymmetry reveals that the Bayesian Lasso is biased towards small shrinkage values ($\lambda$ close to zero), indicating a bias towards zero coefficients. 

Similarly, the Student-t prior also exhibits asymmetry in the distribution over $\xi$, but unlike the Laplace prior, the Student-t prior exhibits a heavier right-hand tail. This preference for large coefficients and having a lower probability for small coefficients is evident from the heavy tails and the absence of a peak at the origin in the distribution plot over $\beta$. In contrast, the horseshoe prior places equal preference on both large and small coefficients, as its distribution over $\xi$ is symmetric.

It is important to note that the Bayesian Lasso is not a special case of the NBP prior, and approximation between the two is challenging due to their distinct tail behaviors. The Lasso (exponential) prior vanishes at a log-super-linear rate in the right tail (i.e. exponentially decaying tails in the original parameterization) :
\begin{equation}
\lim_{\xi\to -\infty} \frac{- \log p_{\rm L}(\xi)}{\xi} = -2 \;\;\;\;\;\;{\rm and}\;\;\;\;\;\;\lim_{\xi\to \infty} \frac{- \log p_{\rm L}(\xi)}{\xi} = \infty
\end{equation}
The beta prime prior, on the other hand, has log-linear tails:
\begin{equation}
\lim_{\xi\to -\infty} \frac{- \log p_{\rm bp}(\xi|a,b)}{\xi} = -2a \;\;\;\;\;\;{\rm and}\;\;\;\;\;\;\lim_{\xi\to \infty} \frac{- \log p_{\rm bp}(\xi|a,b)}{\xi} = 2b.
\end{equation}
This suggests that the beta prime prior can only generalize to priors with heavy tails (log-linear decaying tails).

Let $x^2|\nu, b \sim {\rm IG}(b, 1/\nu)$ and $\nu | a \sim {\rm IG}(a,1)$, then $p(x) \propto x^{2a-1} (1+x^2)^{-a-b}$ \cite{schmidt2019bayesian} whereby marginally, $x$ follows (\ref{eq: lambda_beta_prime}). Using this decomposition, we can express the NBP prior as a scale mixture representation:
\begin{align}\label{eq: GHS_IGIG}
\begin{split}
  \beta_{j}|\tau^2, \lambda_{j}^2, \sigma^2 \; &\sim \; N\left(0, \; \tau^2 \lambda^2_{j} \sigma^2\right) \\
  \lambda^2_{j}|\nu_{j} \; &\sim \; {\rm IG}(b, 1/\nu_{j})\\
  \nu_{g1}, \cdots \nu_{p} \; &\sim \; {\rm IG} (a, 1).
\end{split}
\end{align}
This hierarchy can be viewed as a Student-t prior with an additional IG hyperprior. Upon integrating out the local shrinkage parameter $\lambda$, it results in a Student-t prior distribution over the regression coefficient $\beta_j$ of the form
\begin{align}\label{eq: Student-t}
\begin{split}
  p(\beta_j| \tau^2, \sigma^2, \nu, b) &= \int p(\beta_j| \tau^2, \sigma^2, \lambda^2)p(\lambda^2|\nu, b) \,d\lambda^2\\
  &=  \frac{\sqrt{\nu_j} \, \Gamma(b + \frac{1}{2})}{\sqrt{2 \pi \sigma^2 \tau^2} \, \Gamma(b)} \left( 1 + \frac{\beta_j^2 \nu_j}{2 \sigma^2 \tau^2} \right)^{-\left(b + \frac{1}{2}\right)}
\end{split}
\end{align}
Consequently, the NBP model can be seen as the product of independent Student-t distributions over $\beta_j$.

\begin{figure}[t]
     \centering
     \begin{subfigure}[b]{0.45\textwidth}
         \centering
         \includegraphics[width=\textwidth]{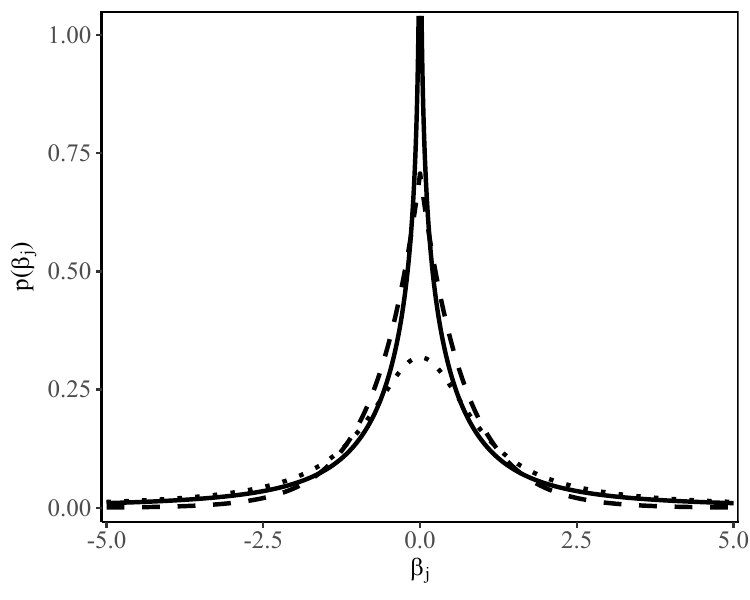}
         \caption{}\label{fig: LassoVSHS(a)}
     \end{subfigure}
     \hspace{1em}%
     \begin{subfigure}[b]{0.45\textwidth}
         \centering
         \includegraphics[width=\textwidth]{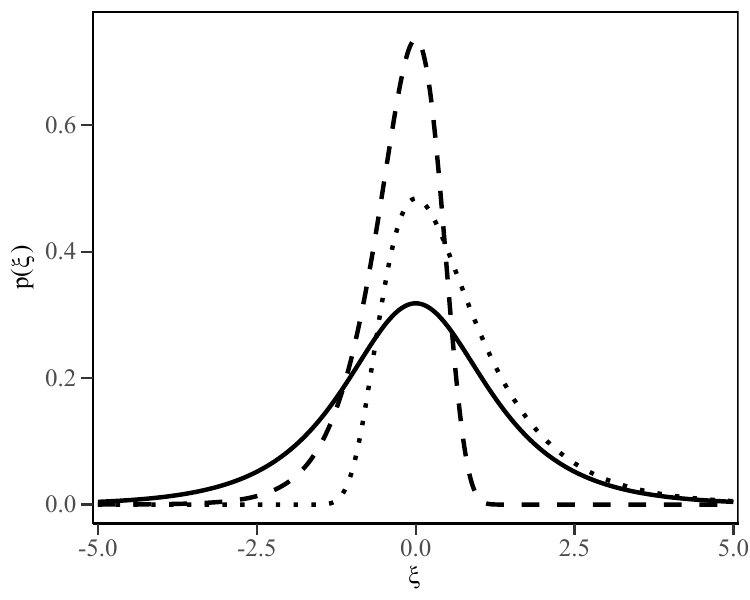}
         \caption{}\label{fig: LassoVSHS(b)}
     \end{subfigure}
     \caption{Comparison of the Lasso (dashed), Student-${\rm t_{\gamma = 1}}$  (dotted), and horseshoe (solid) densities on the marginal prior over $\bm \beta$ (left) and transformed prior for the local shrinkage parameter with ${\bm \xi} = \log {\bm \lambda}$ (right).} \label{fig: LassoVSHS}
\end{figure}

\section{GRASP model} \label{sec:Grouped GHS model}
Using the hierarchical Bayesian grouped regression model \eqref{eq: Bayesian_group_regression_hierarchy} - \eqref{eq: group_shrinkage_hier}, we assign the beta prime prior to both the group and local shrinkage parameters. We assume no prior knowledge on the degree of shrinkage of the regression coefficients and assign the recommended half-cauchy prior for the global variance parameter~\cite{gelman2006prior,polson2012half}:
\begin{align}\label{hier: betaP_prior}
\begin{split}
  \lambda^2_{gj} \; &\sim \; B'(a_g,b_g) \\
  \delta^2_g \; &\sim \; B'(a, b) \\
  \tau \; &\sim \; C^+(0,1).
\end{split}
\end{align}
By selecting appropriate values for $a_g$ and $b_g$, we can control the level of adaptivity within each individual group, while $a$ and $b$ govern the adaptivity of the overall group shrinkage parameter. The introduction of the group shrinkage parameter $\delta$ introduces apriori correlation among the local shrinkage parameters within the same group. Proposition \ref{propo:corr_delta_lambda} provides the exact formula for computing this correlation. It reveals that a higher variance in the sequence $\lambda_g$ corresponds to a lower correlation, indicating that $\delta$ has a lesser impact on the local shrinkage parameter within group $g$. Conversely, a higher variance in $\delta$ results in a higher correlation, indicating that $\delta$ exerts a greater influence on the values of $\lambda$ within the same group. With reference to Figure \ref{fig:group_local_illustration}, this implies that if, for instance, group 1 contains more zero coefficients compared to group 2, we can enhance the shrinkage of the $\lambda$ parameters in group 1 towards zero by adjusting the values of $\delta_1$ without affecting the estimates in group 2.

\begin{figure}[t]
\setlength{\unitlength}{0.14in} 
\centering 
\begin{picture}(32,16) 
\put(-3,14){\footnotesize Regression coefficients}
\put(8,13){\framebox(2.5,2.5){$\beta_1$}}
\put(10.5,13){\framebox(2.5,2.5){$\beta_2$}}
\put(13,13){\framebox(2.5,2.5){$\beta_3$}}
\put(15.5,13){\framebox(2.5,2.5){$\beta_4$}}
\put(18,13){\framebox(2.5,2.5){$\beta_5$}}
\put(20.5,13){\framebox(2.5,2.5){$\beta_6$}}
\put(23,13){\framebox(4.5,2.5){$\cdots$}}
\put(27.5,13){\framebox(2.5,2.5){$\beta_{p-2}$}}
\put(30,13){\framebox(2.5,2.5){$\beta_{p-1}$}}
\put(32.5,13){\framebox(2.5,2.5){$\beta_p$}}
\put(-3,11){\footnotesize Local shrinkage parameters}
\put(8,10){\framebox(2.5,2.5){$\lambda_1$}}
\put(10.5,10){\framebox(2.5,2.5){$\lambda_2$}}
\put(13,10){\framebox(2.5,2.5){$\lambda_3$}}
\put(15.5,10){\framebox(2.5,2.5){$\lambda_4$}}
\put(18,10){\framebox(2.5,2.5){$\lambda_5$}}
\put(20.5,10){\framebox(2.5,2.5){$\lambda_6$}}
\put(23,10){\framebox(4.5,2.5){$\cdots$}}
\put(27.5,10){\framebox(2.5,2.5){$\lambda_{p-2}$}}
\put(30,10){\framebox(2.5,2.5){$\lambda_{p-1}$}}
\put(32.5,10){\framebox(2.5,2.5){$\lambda_p$}}
\put(8,9.8){$\underbrace{\hspace*{9.9em}}$}
\put(11.8,8.2){$a_1, b_1$}
\put(18.2,9.8){$\underbrace{\hspace*{4.8em}}$}
\put(19.4,8.2){$a_2, b_2$}
\put(28,9.8){$\underbrace{\hspace*{6.8em}}$}
\put(30,8.2){$a_G, b_G$}
\put(-3,6){\footnotesize Group shrinkage parameters}
\put(8,5){\framebox(10,2.5){$\delta_1$}}
\put(18,5){\framebox(5,2.5){$\delta_2$}}
\put(23,5){\framebox(7,2.5){$\cdots$}}
\put(30,5){\framebox(5,2.5){$\delta_G$}}
\put(8,4.8){$\underbrace{\hspace*{27em}}$}
\put(20.6,3.2){$a, b$}
\put(-3,1.2){\footnotesize Global shrinkage parameter}
\put(8,0.2){\framebox(27,2.5){$\tau$}}
\end{picture}
\captionsetup{font = footnotesize}
\caption{An illustration of how the shrinkage parameters and hyperparameters interact. Each regression coefficient is associated with a local shrinkage parameter $\lambda$ that controls individual shrinkage, while the global shrinkage parameter $\tau$ controls the overall shrinkage. Within each group of local shrinkage parameters, an additional group shrinkage parameter $\delta$ further amplifies or diminishes the shrinkage of $\lambda$ values within that specific group. Each group of $\lambda$ is also characterized by its unique hyperparameters, denoted as $a_g$ and $b_g$, which control the adaptivity and level of sparsity of the $\lambda$s. An additional pair of hyperparameters $a$ and $b$ is introduced to control the adaptivity of $\delta$.} 
\label{fig:group_local_illustration} 
\end{figure}
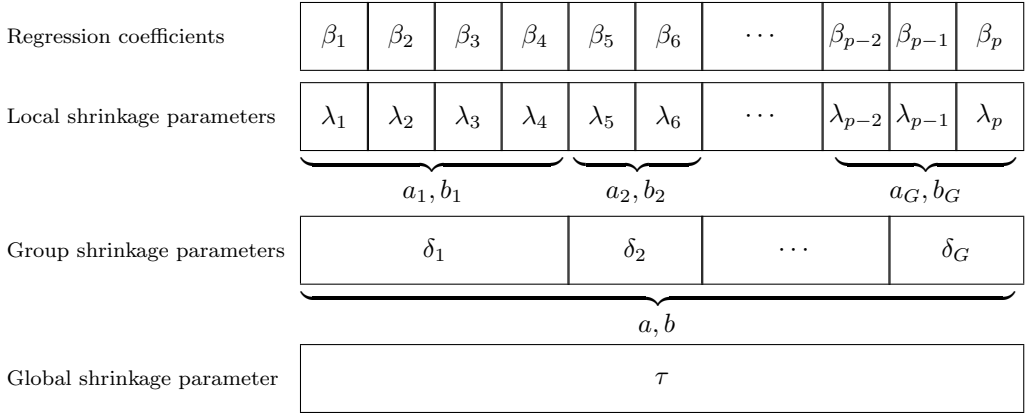

\begin{proposition} \label{propo:corr_delta_lambda}
The correlation between the log-transformed product of the group shrinkage parameter $\delta_g$ and the local shrinkage parameters $\lambda_g$ within the same group, $g$ can be computed as the ratio of the variance of $\delta$ to the sum of the variance of $\delta$ and $\lambda_g$: 
\begin{equation}
    {\rm corr}(\log(\delta_g\lambda_{gi}), \log(\delta_g\lambda_{gj})) = \frac{\V{}{\log(\delta)}}{ \V{}{\log(\delta)} + \V{}{\log(\lambda_g)} }
\end{equation}
where $(\lambda_{gi}, \lambda_{gj}) \in {\bm \lambda_g}$.
\end{proposition}

\begin{proof}
See Appendix \ref{apx: prop_proof}
\end{proof}

\noindent The variance of the local ($\lambda$) and group ($\delta$) shrinkage parameters on logarithmic scales is influenced by the tail behavior of their prior distributions in the original space. Lighter tails in the original space result in lower variance in the log scale and conversely, a heavier tail leads to higher variance. Specifically, the hyperparameters $a$ and $b$ in the beta prime priors~\eqref{hier: betaP_prior} control the tail behavior of both, $\lambda$ and $\delta$, thereby controlling the extent of the grouping effect. Proposition~\ref{propo:corr_delta_lambda} suggests that when the variance of $\log \delta$ is large relative to the variance of $\log \lambda$, the grouping effect becomes the dominant factor in the hyperparameter estimation (i.e., the correlation of the shrinkage parameters within the group is highest and the effect of $\delta$ becomes more significant). Conversely, if the variance of $\log \delta$ is small relative to the variance of $\log \lambda$ then the individual local shrinkage parameters will exhibit greater a priori variability. Therefore, by adjusting the beta prime tail behavior, we can control the {\em a priori} correlation among the shrinkage parameters. A significantly larger variance in $\log \delta$ (relative to the variance in $\log \lambda$) suggests strong grouping effect, whereas a smaller variance in $\log \delta$ (relative to the variance in $\log \lambda$) implies weaker impact.

We compare the exact prior correlation values of different group shrinkage priors in Table \ref{tab:corr_compare}. For the non-adaptive grouped horseshoe method, the apriori correlation between the effective local shrinkage parameters within a group is fixed at 0.5. This fixed correlation can be recovered by setting $a$ and $b$ to 0.5 in both the GIGG and GRASP settings. The introduction of hyperparameters $a$ and $b$ offers the flexibility to control not only the sparsity levels between and within groups but also the prior correlation by appropriately selecting priors for these parameters. Through Proposition~\ref{propo:corr_delta_lambda}, we quantify how the choice of priors for the shape parameters of the beta prime distribution ($a, b$) influences the correlation among local shrinkage parameters within a group. This novel capability enables a systematic exploration of the prior correlation structure, addressing a previously unexamined aspect of group shrinkage models. Figure~\ref{fig:correlation_behavior} illustrates how different prior choices result in distinct distributions of the apriori correlation. We explore four distinct behaviors, detailed below:
\begin{enumerate}
    \item \textbf{Proposed prior - $a, b, a_g, b_g \; \sim \; C^+(0,1)$ (Figure \ref{fig: corr_all_halfCauchy}).} \\
    In section \ref{sec: sample ab}, we proposed placing a half Cauchy prior on the beta prime shape parameters and sample from the conditional posterior to estimate the parameters and learn the level of sparsity and signal from the data. Figure \ref{fig: corr_all_halfCauchy} shows that under this prior configuration, both GIGG and GRASP yield similar correlation distributions, with GRASP showing less concentration at the extremes (correlation of 1 and 0) and more probability mass in the middle range compared to GIGG. This behavior aligns with the intuition that models with intermediate correlation values are often indistinguishable in practice, while models with either no group shrinkage (only individual-level effects) or strong group shrinkage (minimal individual variation) are the most distinguishable. This prior structure effectively encodes the likelihood of two scenarios: (1) no group effect (only local individual effects) and (2) only group effects with no local effects, making it an intuitive and interpretable prior choice.
    
    \item \textbf{Non-adaptive grouped horseshoe - $a^q, b^q, a^q_g, b^q_g \; \sim \; C^+(0,1)$ with $q = 10^4$ (Figure \ref{fig: corr_all_halfCauchyQ})}. \\
    By choosing the hyperparameters' priors to be the $q$-th root of a half-Cauchy distribution with $q$ set to a sufficiently large value (e.g., $10^4$), we recover a point mass at 0.5 for the correlation. This result replicates the behavior of the non-adaptive grouped horseshoe.
    \item \textbf{Uniform aprior correlation - $a, b, a_g, b_g \; \sim \; C^+(0,25)$ (Figure \ref{fig: corr_uniform})}. \\
    This prior configuration seems to be expressing no explicit preference for any specific correlation value, given its approximately uniform distribution across the correlation range. Notably, this behavior is achievable only with GRASP; GIGG, under the same prior settings, exhibits residual biases toward correlations near 0 and 1 due to its structural limitations. While a uniform prior correlation may seem like a reasonable default for expressing no apriori preference, it is generally discouraged. A uniform prior does not necessarily encode prior ignorance about the correlation coefficient, as the correlation parameter operates on a nonlinear scale. Although the parameter itself might be defined on a linear scale, its effects are not linear in practice. Instead, more principled alternatives, such as the proposed prior or Jeffrey’s prior, which accounts for Fisher information, are generally preferred as they better capture the behavior of the parameter's effective distribution.

    \item \textbf{Unique trimodal behavior - $a, b \sim C^+(0,1); a^q_g, b^q_g \; \sim \; C^+(0,1)$ with $q = 10^4$ (Figure \ref{fig: corr_half_halfCauchy_halfQ})}. \\
    This configuration highlights the flexibility of the GRASP hierarchy compared to GIGG. Figure \ref{fig: corr_half_halfCauchy_halfQ} reveals a trimodal correlation density for GRASP, with modes at:(i) zero - group shrinkage dominates, with negligible/non-existence individual shrinkage; (ii) one - individual shrinkage dominates, with negligible/ non-existence group shrinkage; (iii) 0.5 - Equal contribution from both shrinkage levels; with non-zero probabilities across other correlation values. This trimodal behavior cannot be achieved using GIGG due to its limited degrees of freedom. The GIGG hierarchy requires $a = b$ for symmetric correlation densities, restricting independent control over concentration and tail behavior. Consequently, in the GIGG structure, any effort to control the probability weight at the origin or tails inherently results in an implied fixed correlation. In contrast, GRASP’s additional degrees of freedom (extra $a$ and $b$ hyperparameter) allow simultaneous control over fixed correlation levels and the weighting of tail probabilities, enabling far greater flexibility in modeling correlation behaviors.
\end{enumerate}

\begin{table*}[t]
\caption{Comparison of the correlation between the effective shrinkage parameters within a group across different hierarchy configurations. $\phi'(\cdot)$ denotes the trigamma function. Detailed derivations are provided in Appendix \ref{apx: prop_GIGG} and \ref{apx: prop_GNBP}.}\label{tab:corr_compare}
\vspace{-2mm}
\begin{center}
\begin{tabular}{@{}c@{\hskip 0.1in}c@{\hskip 0.1in}c@{\hskip 0.1in}c@{}}
\toprule
Group Horseshoe & GIGG & Group Horseshoe+ & GRASP \\
\midrule
0.5 & $\displaystyle \frac{\phi'(a_g)}{\phi'(a_g) + \phi'(b_g)}$ & 0.5 & $\displaystyle \frac{\phi'(a) + \phi'(b)}{\phi'(a) + \phi'(b) + \phi'(a_g) + \phi'(b_g)}$\\
\bottomrule
\end{tabular}   
\end{center}
\vspace{-3mm}
\end{table*}
\begin{figure}[t]
     \centering
     \begin{subfigure}[b]{0.48\textwidth}
         \centering
         \includegraphics[width=\textwidth]{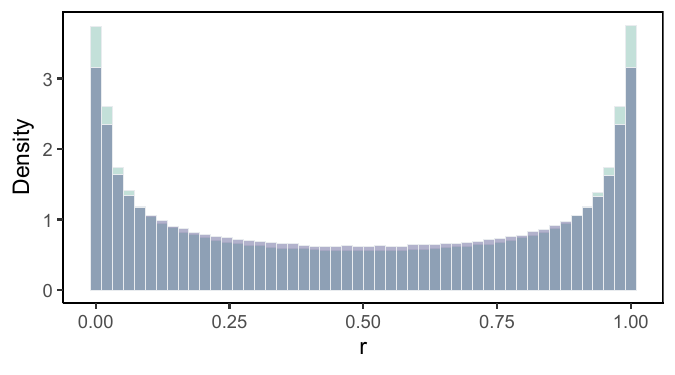}
         \caption{$a, b, a_g, b_g \; \sim \; C^+(0,1)$} \label{fig: corr_all_halfCauchy}
     \end{subfigure}
     \hspace{1em}%
     \begin{subfigure}[b]{0.48\textwidth}
         \centering
         \includegraphics[width=\textwidth]{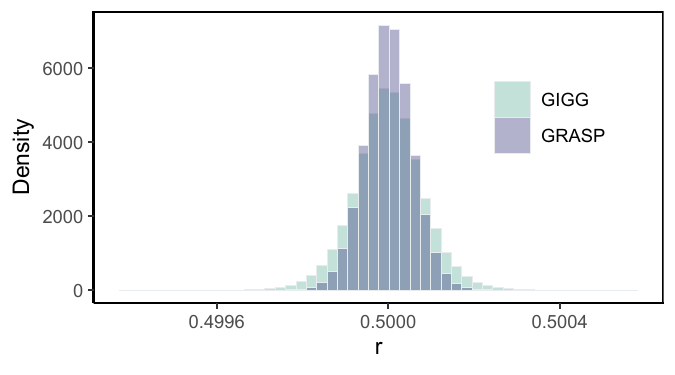}
         \caption{$a^q, b^q, a^q_g, b^q_g \; \sim \; C^+(0,1)$} \label{fig: corr_all_halfCauchyQ}
     \end{subfigure}
     \begin{subfigure}[b]{0.48\textwidth}
         \centering
         \includegraphics[width=\textwidth]{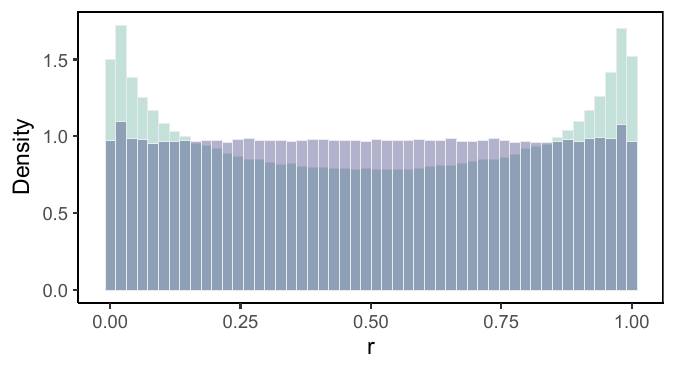}
         \caption{$a, b, a_g, b_g \; \sim \; C^+(0,25)$} \label{fig: corr_uniform}
     \end{subfigure}
     \hspace{1em}%
     \begin{subfigure}[b]{0.48\textwidth}
         \centering
         \includegraphics[width=\textwidth]{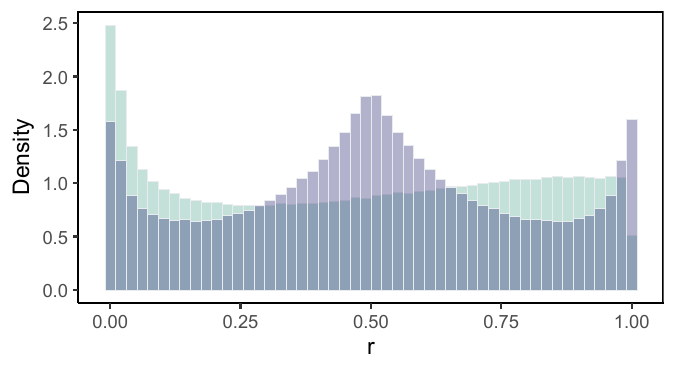}
         \caption{$a^q, a^q_g \; \sim \; C^+(0,1); b, b_g \; \sim \; C^+(0,1)$}\label{fig: corr_half_halfCauchy_halfQ}
     \end{subfigure}
     \caption{Comparison of apriori correlation distribution between GIGG and GRASP under different shape hyperparameter prior configurations. $r = {\rm corr}(\log(\delta_g\lambda_{gi}), \log(\delta_g\lambda_{gj}))$ and $q = 10^4$.}  \label{fig:correlation_behavior}
\end{figure}
\section{Posterior Sampling} \label{sec:posterior sampling}
By adopting the scale mixture representation of the beta prime distribution in (\ref{eq: GHS_IGIG}), we present the revised grouped NBP hierarchy for GRASP as follows
\begin{align}\label{Revised_grouped_regression_hierarchy}
\begin{split}
  \bm{y}|\bm{X}, \bm{\beta}, \sigma^2 \; &\sim \; N_n\left(\bm{X\beta}, \; \sigma^2\bm{I}_n\right) \\
  \sigma^2 \; &\sim \; \sigma^{-2} d\sigma^2  \\
  \beta_{gj}|\tau^2, \delta^2_g, \lambda_{gj}^2, \sigma^2 \; &\sim \; N\left(0, \; \tau^2 \lambda^2_{gj} \delta^2_g \sigma^2\right) \\
  \lambda^2_{gj}|\nu_{gj} \; \sim \; {\rm IG}(b_g, 1/\nu_{gj}) \;\;\;\; &, \;\;\;\; \nu_{g1}, \cdots \nu_{gp} \; \sim \; {\rm IG} (a_g, 1)\\
  \delta^2_g|\zeta_g \; \sim \; {\rm IG}(b, 1/\zeta_g) \;\;\;\; &, \;\;\;\; \zeta_{1}, \cdots \zeta_{G} \; \sim \; {\rm IG} (a, 1)\\
  \tau \; &\sim \; C^+(0,1).
\end{split}
\end{align}
Using this hierarchy, we develop a Gibbs sampler that iteratively samples from the following full conditional densities:
\begin{align}\label{hier: gibbs_sampler}
\begin{split}
  {\bm \beta} | \cdot \; \sim \; N_p &({\bf A}^{-1}{\bf X}^{\rm T}{\bf y}, \sigma^2{\bf A}^{-1}) \\
  \sigma^2 | \cdot \; \sim \; {\rm IG} &\left(\frac{n+p}{2}, \frac{||{\bf y} - {\bf X}\bm{\beta}||^2 + {\bm \beta}^{\rm T} (\tau^2 {\bm \Lambda})^{-1} {\bm \beta}}{2} \right) \\
  \lambda_{gj}^2|\cdot \; \sim \; {\rm IG} \left(\frac{1}{2} + b_g, \;\; \frac{1}{\nu_{gj}} + \frac{\beta_{gj}^2}{2\tau^2\sigma^2} \right) \;\;\;\; &, \;\;\;\;  \nu_{gj} | \cdot \; \sim \; {\rm IG} \left(a_g + b_g, \;\; 1 + \frac{1}{\lambda_{gj}^2} \right)\\
  \delta_{g}^2|\cdot \; \sim \; {\rm IG} \left(\frac{ng}{2} + b, \;\; \frac{1}{\zeta_g} + \sum_{g = 1}^{ng}\frac{\beta_{g}^2}{2\tau^2\sigma^2\lambda^2_g} \right) \;\;\;\; &, \;\;\;\;  \zeta_g | \cdot \; \sim \; {\rm IG} \left(a + b, \;\; 1 + \frac{1}{\delta_g^2} \right)\\
\end{split}
\end{align}
where $ng$ is the number of predictors in group $g$, ${\bf A} = {\bf X}^{\rm T} {\bf X} +  (\tau^2 {\bm \Lambda})^{-1}$ and ${\bm \Lambda} = {\rm diag}(\lambda^2_1, \cdots, \lambda^2_p)$.  In the absence of a group structure, this Gibbs sampler can be readily transformed into a normal linear regression with a beta prime prior by setting $g=1$ and not sampling the $\delta$ and $\zeta$ parameters.

\subsection{Estimating the adaptive parameters $a$ and $b$} \label{sec: sample ab}
To allow the NBP prior adapt to varying sparsity and signal sizes, we estimate the shape hyperparameters $a$ and $b$ from the data by placing hyper-priors on them and sampling from their conditional posterior distributions. We present an efficient algorithm to approximate the conditional posterior of the beta shape parameters $a$ and $b$ using gamma distributions. Given observations $x_1, \cdots, x_n | a, b \sim {\rm Beta}(a, b)$ and a prior for the shape parameters $a$ and $b$, the goal is to approximate the conditional distributions $p(a | x_1, \cdots, x_n, b)$ and $p(b | x_1, \cdots, x_n, a)$ using gamma distributions:
\begin{equation}
    p(a | x_1, \cdots, x_n, b) \approx {\rm Gamma(k_a, \theta_a)}, \;\;\;\;\;\;  p(b | x_1, \cdots, x_n, a) \approx {\rm Gamma(k_b, \theta_b)}\\
\end{equation}
A half-Cauchy prior is assigned to the shape parameters:
\begin{equation}\label{hier: ab_hyperprior}
    a_1\cdots a_G, b_1 \cdots b_G, a, b \; \sim \; C^+(0,1). \\
\end{equation}
The prior choice is motivated by the three distinct mode behaviors of the beta distribution: 
\begin{enumerate}
    \item When $a=b=1$, the beta distribution becomes uniform, and any value in the domain is the mode.
    \item For $a,b > 1$, the beta distribution has a unique mode at $(a-1)/(a+b-2)$.
    \item When $a <1$ or $b<1$, the distribution diverges, exhibiting a pole at the origin ($a<1$) or at 1 ($b<1$), with no well-defined mode.
\end{enumerate}
The half-Cauchy prior assigns equal probabilities to $a < 1$ and $a>1$ (similarly for $b$), reflecting a prior belief in having an equal likelihood of observing super-efficiency (a pole at the origin) or not. This symmetry assumes no apriori preference for the degree of sparsity, which is particularly useful for robust sampling. Furthermore, the half-Cauchy prior is widely recognized as a robust default for scale parameters due to its desirable theoretical properties (such as minimaxity and many others as documented in \citet{maruyama2024minimaxity,polson2012half}), making it a practical and theoretically sound choice for this application.

Our proposed algorithm, inspired by \citet{salimans2013fixed}, leverages Metropolis-Hastings sampling, where the proposal distribution (a gamma distribution) is expressed in exponential family form with natural parameters. We then solve for the gamma distribution parameters that best approximate the target distribution. If the proposal distribution exactly matches the target distribution, this method provides an exact fit. Otherwise, it yields an approximation that minimizes the Kullback-Leibler divergence between the proposal and target distributions \cite{salimans2013fixed}. This Metropolis-Hastings algorithm seamlessly integrates into the Gibbs sampling framework \eqref{hier: gibbs_sampler}. Further details of the proposed algorithm are provided in Appendix \ref{apx: Fast and Accurate Approximation of the Beta Shape Parameters}.

\subsection{Estimating the global shrinkage parameter $\tau$}

We explore two approaches for sampling the global shrinkage parameter $\tau$: 1) assigning $\tau$ a standard half-Cauchy prior $C^+(0,1)$; 2) sampling from a truncated half-Cauchy prior that restricts $\tau$ to the interval $[1/n, 1]$, as suggested by \cite{van2017adaptive}. The first approach involves adopting the inverse gamma scale mixture representation of the half-Cauchy prior \cite{makalic2015simple} such that:  
\begin{align}
\begin{split}
  \tau^2 \, | \, \omega \; &\sim \; {\rm IG (1/2, 1/\omega)},\\
  \omega \; &\sim \; {\rm IG (1/2, 1)},
\end{split}
\end{align}
and sampling $\tau$ simply requires sampling from the following conditional posterior distribution:
\begin{align}
\begin{split}
  \tau^2 \, | \, \cdot \; &\sim \; {{\rm IG} \left(\frac{p+1}{2}, \frac{1}{\omega} + \frac{1}{2\sigma^2}\sum_{j = 1}^p \frac{\beta_j^2}{\lambda_j^2} \right)},\\
  \omega \,|\, \cdot \; &\sim \; {{\rm IG} \left(1, 1 + \frac{1}{\tau^2} \right)}.
\end{split}
\end{align}
The posterior mean estimates of $\tau$ are computed as the average of the $\tau$ samples drawn from this conditional posterior distribution.

In the second approach, we adopted a modified rejection sampling scheme, following the framework introduced by \citet{schmidt2019bayesian}. We conducted numerical experiments under settings similar to those described in Section~\ref{sec:Results}, and observed that truncating the global shrinkage parameter $\tau \in [1/n, 1]$, as recommended by \citet{van2017adaptive}, resulted in worst predictive performance in the context of grouped linear regression models. Consequently, in the results section, we only present results without truncation of $\tau$. While truncating $\tau$ might offer theoretical guarantees in the normal means model, its use does not appear to be justified for grouped regression models. Previous work has shown that any $\tau \in [1/n, 1]$ can lead to suboptimal contraction rates for sparse high-dimensional regression, when $p \gg n$ \cite{bai2021beta, song2023nearly}.


\section{Simulated Results}\label{sec:Results}

We performed a comparative evaluation of GRASP against existing NBP estimators for linear regression and the GIGG model for grouped linear regression. All experiments were conducted using the R statistical platform. The datasets and code used for the experimental results in this section are publicly available\footnote{Available at \url{https://github.com/shuyu-tew/GRASP.git}}. The implementation of the GRASP will be included in \textbf{bayesreg}\footnote{Available at \url{https://cran.r-project.org/web/packages/bayesreg/index.html}} the Bayesian regression package. For the existing state-of-the-art NBP estimator, we used the \texttt{nbp} function from the \texttt{NormalBetaPrime} R package, and for the GIGG model, we used the \texttt{gigg} R package. Unless specified otherwise, all function arguments were set to their default values. Without any loss of generality, all predictors and the target variable are standardized to have means zero and unit variance.

In terms of prior configuration, GIGG deviates slightly from our approach by adopting a different setting where the product of the group shrinkage parameter $\delta$ and the local shrinkage parameter $\lambda$ follows a beta prime distribution. This is achieved through the scale mixture representation of the beta prime distribution, for which given two random variables $x$ and $s$ such that
\begin{align}
\begin{split}
  x|\theta \; \sim \; {\rm G} (a, z)\;\;\;\; {\rm and} \;\;\;\; \theta \; \sim \; {\rm IG} (b, z);
\end{split}
\end{align}
then $x | \theta \sim B'(a,b)$. Using this decomposition, we can directly compare the prior employed by GIGG (left) and our own (right):
\begin{align}
\begin{split}
  \lambda_{gj}^2 \; \sim \; {\rm IG} (b_g, 1)\;\;\;\; &\;\;\;\; \lambda_{gj}^2 \; \sim \; B' (a_g, b_g)\\
  \delta_g^2 \; \sim \; {\rm G} (a_g, 1) \;\;\;\; &\; \;\;\;\; \delta_g^2 \; \sim \; B' (a, b)\\
\end{split}
\end{align}
In \citet{boss2023group}, they established $a_g = 1/n$ and learn $b_g$ through marginal maximum likelihood estimation (MMLE). This implies a certain alignment with our approach, suggesting a link between our group shrinkage hyperparameter's $a$ and their $a_g$, given that their $a_g$ controls the degree of shrinkage in $\delta$. Consequently, as demonstrated in Table \ref{tab:grouped_regression_results}, when conducting the group regression experiment to compare against GIGG, we provide two distinct configurations for our method. The first configuration involves setting $a = 1/n$ and only allowing the estimation of $b$. In the second scenario, both $a$ and $b$ are allowed to be estimated using the mgrad algorithm. In both cases, the hyperparameters for $\lambda^2$, ($a_g, b_g$) are subject to estimation.

We present four distinct yet related experimental settings, as summarized in Table \ref{tab: experiment details}. The first two experiments are adapted from \citet{boss2023group}. In the first setting, the signal is concentrated in one of the regressors within each of the five groups, while in the second setting, the signal is more evenly distributed across all regressors but limited to only the first group. To further explore different scenarios, we introduce two additional variations: the third experiment involves dense signals (more than half of the regressors with signals) in all six groups, and the fourth experiment features a combination of dense signals in half of the groups and few to no signal regressors in the other half. Pairwise correlations within each group are set to 0.8 and the pairwise correlation across groups is 0.2. In all simulation settings, the residual error variance $\sigma^2$ is selected to achieve varying signal-to-noise ratios, ${\rm SNR} \in (0.2, 1, 5)$ with
\[
    \sigma^2 = \left( \frac{1- {\rm SNR}}{{\rm SNR}} \right) {\bm {\beta}}^{\rm T} \Sigma {\bm {\beta}}.
\]
Following \cite{boss2023group}, we evaluate the estimation properties based on the empirical mean-squared error (MSE) stratified by null and non-null coefficients across 100 replicates. Specifically, we calculate the $\hat{\rm MSE}$ as the average squared difference between the estimated $\hat{\bm \beta}$ and the true $\bm \beta$ for each of the 100 simulated datasets. The true coefficients for each of the experimental settings are given in Table \ref{tab: experiment details}.
{\renewcommand{\arraystretch}{1.4}%
\begin{table}[bt]
\footnotesize
\captionsetup{skip=0pt,belowskip=0pt, font = footnotesize}
\caption{Details of the experimental settings. AV / GS represents the ratio of active variables to the group size, Coef is the true $\beta$ coefficients assigned to each group and U($v$) denotes uniform sampling from the vector $v$, where the size corresponds to the number of active variables within each group.}
\vskip -0.1in
\label{tab: experiment details}
\begin{center}
\begin{tabular}{cC{1.2cm}cC{1.2cm}cC{1.2cm}cC{1.2cm}c}
\toprule
& \multicolumn{4}{c}{$p = 50, n = 500$} & \multicolumn{2}{c}{$p = 100, n = 300$}& \multicolumn{2}{c}{$p = 80, n = 200$}\\ 
\cmidrule(lr){2-5}\cmidrule(lr){6-7}\cmidrule(l){8-9}
& \multicolumn{2}{c}{Concentrated} & \multicolumn{2}{c}{Distributed} & \multicolumn{2}{c}{Dense} & \multicolumn{2}{c}{Half}\\ 
\cmidrule(lr){2-5}\cmidrule(lr){6-7}\cmidrule(l){8-9}
Group & AV / GS  & Coef & AV / GS  & Coef & AV / GS  & Coef & AV / GS  & Coef\\
\cmidrule(r){1-1}\cmidrule(lr){2-5}\cmidrule(lr){6-7}\cmidrule(l){8-9}
1 & 1/10 & 0.5 & 10/10 & 0.5 & 27/30 & [3,5,8] & 22/25 & 0.8\\
2 & 1/10 & 1 & 0/10 & - & 8/10 & [2,3,4,5] & 0/10 & -\\
3 & 1/10 & 1.5 & 0/10 & - & 18/20 & 7.5 & 3/10 & 2.5\\
4 & 1/10 & 2 & 0/10 & - & 4/5 & $[9.5, 8, 7]$ & 8/10 & 1.5\\
5 & 1/10 & 2 & 0/10 & - & 14/15 & 1.5 & 0/5 & -\\
6 & {---} & {---} & {---} & {---} & 18/20 & 0.5 & 4/15 & $[1,2,3,5]$\\
7 & {---} & {---} & {---} & {---} & {---} & {---} & 1/5 & 2\\
\hline
\end{tabular}
\end{center}
\vskip -0.1in
\end{table}}
The results for both standard linear regression and grouped linear regression are presented in Table \ref{tab:linear_regression_results} and Table \ref{tab:grouped_regression_results} respectively. The experimental settings for both scenarios remain consistent, except that in standard linear regression, grouping information is not passed into the estimator during model fitting. 

\paragraph*{Standard Linear Regression}
To apply GRASP to standard linear regression problems (i.e., without grouping), we exclude the grouped shrinkage parameters from the model - specifically, we do not sample the group-level shrinkage parameters $\delta$ and $\zeta$ in hierarchy \eqref{Revised_grouped_regression_hierarchy}. We refer to this version of the model as RASP (Regression with Adaptive Shrinkage Priors). In this experiment, we compare the RASP against the NBP estimator implemented in the \texttt{nbp} R function. For each estimator, 10,000 samples were drawn from the corresponding posterior distribution with a burnin period of 10,000 samples. Given the absence of a thinning option in the \texttt{nbp} R function, and to ensure fair time complexity comparisons, we set the thinning level of RASP to be 1 (i.e. no thinning). The results in Table \ref{tab:linear_regression_results} suggest that the NBP estimator consistently performs worse than RASP, particularly when the SNR is low. Its performance deteriorates further when the true regression coefficients are not sparse, offering only marginal improvements in mean squared error (MSE) over the least squares method. Nevertheless, in settings with high SNR, the NBP estimator becomes comparable to RASP, although the latter maintains a slight advantage in terms of prediction accuracy. Overall, RASP consistently provides superior prediction accuracy when compared to the existing state-of-the-art NBP estimator, with the additional advantage of being up to six times faster.


\begin{table}[bt]
\scriptsize
\captionsetup{skip=0pt,belowskip=0pt,font = scriptsize}
\caption{\small Mean-squared errors (MSE) for the standard linear regression problem with \textbf{no grouping} information provided during model fitting. Z0 is the MSE for estimating the null coefficients, NZ0 is the MSE for estimating the null coefficients, and ${\rm OA} = {\rm Z0} + {\rm NZ0}$ is the overall MSE. Cells in bold highlight the estimator with the lowest overall MSE, one for each distinct signal-to-noise ratio (SNR) in the problem. The results for three estimators are presented - the ordinary least squares estimator (OLS), the normal-beta prime estimator using the \texttt{nbp} R function (NBP), and RASP.}
\label{tab:linear_regression_results}
\begin{center}
\begin{sc}
\begin{tabular}{@{}l@{}l@{\hskip 0.03in}c@{\hskip 0.05in}c@{\hskip 0.05in}c@{\hskip 0.03in}c@{\hskip 0.06in}c@{\hskip 0.05in}c@{\hskip 0.05in}c@{\hskip 0.03in}c@{\hskip 0.06in}c@{\hskip 0.05in}c@{\hskip 0.05in}c@{\hskip 0.03in}c@{\hskip 0.06in}c@{\hskip 0.05in}c@{\hskip 0.05in}c@{\hskip 0.03in}c@{}}
\toprule
& & \multicolumn{4}{c}{Concentrated} & \multicolumn{4}{c}{Distributed}& \multicolumn{4}{c}{Half} & \multicolumn{4}{c}{Dense}\\
\cmidrule(lr){3-6} \cmidrule(lr){7-10} \cmidrule(lr){11-14}\cmidrule(l){15-18}
SNR &  & Z0 & NZ0 & OA & Time & Z0 & NZ0 & OA & Time & Z0 & NZ0 & OA & Time & Z0 & NZ0 & OA & Time \\
\midrule
\multirow{3}{*}{0.2} & OLS  & 212.8 & 25.5 & 238.3 & 0.01 & 467.2 & 112.9 & 580.2 & 0.01 & 35466 & 33141 & 68608 & 0.01 & 91740 & 285864 & 377604 & 0.01\\
& NBP             & 18.92 & 8.92 & 27.84 & 32.6 & 432.9 & 105.1 & 538.0 & 32.4 & 34878 & 32445 & 67323 & 62.6 & 93328 & 260839 & 354167 & 104 \\
& RASP        & 2.79 & 8.81 & \textbf{11.6} & 8.82 & 5.96 & 5.61 & \textbf{11.6} & 8.81 & 88.7 & 99.9 & \textbf{188.6} & 10.4 & 163.2 & 765.5 & \textbf{928.8} & 13.6 \\
\\
\multirow{3}{*}{1} & OLS              & 8.52 & 1.02 & 9.53 & 0.01 & 18.69 & 4.52 & 23.21 & 0.01 & 1418 & 1325 & 2744 & 0.01 & 3669 & 11434 & 15104 & 0.01 \\
& NBP           &  1.06 & 1.13 & 2.19 & 32.1  & 3.13 & 3.58 & 6.71 & 32.2 & 1070 & 1000 & 2070 & 67.3 & 3561 & 11113 & 14675 & 99.7\\
& RASP        & 0.89 & 1.26 & \textbf{2.16} & 8.83  & 1.65 & 3.08 & \textbf{4.73} & 8.85 & 44.1 & 74.3 & \textbf{118.4} & 10.5 & 182.7 & 501.3 & \textbf{683.9} & 13.1\\
\\
\multirow{3}{*}{5} & OLS              & 0.34 & 0.04 & 0.38 & 0.01 & 0.75 & 0.18 & 0.93 & 0.01 & 56.75 & 53.03 & 109.8 & 0.01 & 146.8 & 457.4 & 604.2 & 0.01 \\
& NBP             & 0.04 & 0.03 & 0.07 & 32.1 & 0.12 & 0.20 & 0.32 & 35.2 & 12.15 & 25.12 & 37.27 &  66.1 & 41.84 & 174.5 & 216.3 & 95.1\\
& RASP        & 0.02 & 0.02 & \textbf{0.04} & 8.89 & 0.09 & 0.22 & \textbf{0.31} & 9.64 & 12.81 & 22.95 & \textbf{35.76} & 10.3 & 42.14 & 153.2& \textbf{195.3} & 12.4 \\
\bottomrule
\end{tabular}
\end{sc}
\end{center}
\end{table}

\paragraph*{Grouped Linear Regression}
In this experiment, we compare GRASP against the GIGG estimator from the \texttt{gigg} R package. For each estimator, 10,000 samples were drawn from the corresponding posterior distribution with a burnin period of 10,000 samples and a thinning level of 5. The results presented in Table~\ref{tab:grouped_regression_results} suggest that with increasing signal-to-noise ratio (SNR), the predictive performance of all estimators (excluding OLS) tends to level out. However, in scenarios with low SNR, GRASP consistently demonstrates superior overall MSE performance. When SNR is set to 1, GIGG occasionally outperforms GRASP, but we find that by aligning certain settings with their implementation, particularly by setting $a_{\delta}$ to $1/n$, GRASP's performance aligns more closely with GIGG. From a broader perspective, the average predictive performance of both the GIGG estimator and GRASP appears comparable. However, the results do suggest that GRASP has an advantage over the GIGG in low-sparsity, low-SNR settings, achieving more accurate coefficient estimates with lower MSE.
\begin{table}[t]
\scriptsize
\captionsetup{skip=0pt,belowskip=0pt, font = scriptsize}
\caption{Mean-squared errors (MSE) for the grouped linear regression problem. Z0 is the MSE for estimating the null coefficients, NZ0 is the MSE for estimating the null coefficients, and ${\rm OA} = {\rm Z0} + {\rm NZ0}$ is the overall MSE. Cells in bold highlight the estimator with the lowest overall MSE, one for each distinct signal-to-noise ratio (SNR) in the problem. OLS records the results for the ordinary least squares estimator; GIGG records the results using the \texttt{gigg} R function, and the remaining two methods record the results of GRASP under two different settings: the first with a fixed group level hyperparameter $a_\delta = 1/n$, and only $b_\delta$ is learned, and the second setting where both $a_\delta$, $b_\delta$ are estimated.}
\label{tab:grouped_regression_results}
\begin{center}
\begin{sc}
\begin{tabular}{@{}l@{\hskip 0.05in}l@{\hskip 0.1in}c@{\hskip 0.05in}c@{\hskip 0.05in}c@{\hskip 0.1in}c@{\hskip 0.05in}c@{\hskip 0.05in}c@{\hskip 0.1in}c@{\hskip 0.05in}c@{\hskip 0.05in}c@{\hskip 0.1in}c@{\hskip 0.05in}c@{\hskip 0.05in}c@{}}
\toprule
& & \multicolumn{3}{c}{Concentrated} & \multicolumn{3}{c}{Distributed}& \multicolumn{3}{c}{Half} & \multicolumn{3}{c}{Dense}\\
SNR & Method & Z0 & NZ0 & OA & Z0 & NZ0 & OA & Z0 & NZ0 & OA & Z0 & NZ0 & OA \\
\midrule
\multirow{5}{*}{0.2} & GIGG($a_g = 1/n, b_g = ?$)          & 4.26 & 9.37 & 13.63 & 2.38 & 11.73 & 14.11 & 183.3 & 315.2 & 498.5 & 458.5 & 1210 & 1668  \\
\\[-1ex]
& GRASP\\
& $a_\delta = 1/n$, $b_\delta = ?$    & 3.81 & 9.51 & 13.32 & 1.55 & 10.04 & \textbf{11.59} & 150.4 & 293.1 & 443.5 & 380.3 & 1118 & 1498 \\
& $a_\delta = ?$, $b_\delta = ?$      & 3.50 & 9.13 & \textbf{12.63} & 2.97 & 8.70 & 11.67 & 159.9 & 258.9 & \textbf{418.8} & 348.6 & 1030 & \textbf{1379} \\
\\[-1ex]
\hline
\\[-1ex]
\multirow{5}{*}{1} & GIGG($a_g = 1/n, b_g = ?$)          &  0.67 & 1.14 & \textbf{1.81} & 0.119 & 2.449 & \textbf{2.568} & 	78.77 & 219.8 & 298.5 & 232.1 & 585.1 & 817.2 \\
\\[-1ex]
& GRASP\\
& $a_\delta = 1/n$, $b_\delta = ?$    & 0.82 & 1.26 & 2.08 & 0.071 & 2.619 & 2.690 & 74.57 & 193.1 & 267.7 & 206.5 & 514.2 & 720.7 \\
& $a_\delta = ?$, $b_\delta = ?$      & 0.97 & 1.38 & 2.36 & 0.226 & 2.815 & 3.041 & 77.93 & 178.6 & \textbf{256.6} & 194.7 & 498.8 & \textbf{693.5} \\
\\[-1ex]
\hline
\\[-1ex]
\multirow{5}{*}{5} & GIGG($a_g = 1/n, b_g = ?$) & 0.02 & 0.02 & 0.04 & 0.006 & 0.161 & 0.167 & 13.92 & 37.75 & \textbf{51.67} & 43.87 & 121.0 & \textbf{164.9}  \\
\\[-1ex]
& GRASP\\
& $a_\delta = 1/n$, $b_\delta = ?$    & 0.02 & 0.02 & 0.04 & 0.003 & 0.163 & \textbf{0.166} & 13.53 & 39.58 & 53.11 & 39.44 & 128.1 & 167.6 \\
& $a_\delta = ?$, $b_\delta = ?$      & 0.03 & 0.02 & 0.05 & 0.012 & 0.163 & 0.175 & 15.47 & 38.71 & 54.18 & 39.6 & 125.7 & 165.3\\
\bottomrule
\end{tabular}
\end{sc}
\end{center}
\end{table}

\section{Real-world Datasets} \label{sec:realWorldDataset}
We evaluate GRASP on two real-world datasets. The first dataset, sourced from the \texttt{sparsegl} R package, examines trustworthiness in information sources during the COVID-19 pandemic. The second dataset, from the \texttt{BSGS} R package, investigates cross-country factors contributing to the 2008–09 economic crisis. Both datasets exhibit group structure in the predictors, making grouped regression a natural modeling choice. We compare GRASP against two benchmark methods: the grouped inverse gamma gamma (GIGG) estimator and the sparse group lasso (SparseGL) implemented in the \texttt{sparsegl} package. For each dataset, we performed a random 70/30 train-test split and repeated the experiment 100 times. Model performance is assessed using mean squared prediction error (MSPE), with results summarized in Table \ref{tab:realdata_results}.

Overall, the adaptive grouped regression frameworks demonstrate superior predictive performance compared to SparseGL, as expected due to its adaptive shrinkage properties. Among the adaptive grouped estimators, GRASP slightly outperforms GIGG in terms of MSPE, highlighting its effectiveness in real-world applications. Further details on the datasets, experimental setup and results discussion are provided in Sections \ref{sec: trust_results} and \ref{sec: economic_results}, respectively.
\begin{table*}[t]
\caption{Comparing the predictive performance of GRASP, GIGG and SparseGL on real-world datasets. Reported values are the mean squared prediction error averaged across 100 repetition of a 70/30 train test split with the standard error in parenthesis. $n$ is the training size, $p$ is the number of predictors including the bspline expansion and dummy variable transformations. Bold values indicate the best-performing model for each dataset.}\label{tab:realdata_results}
\vspace{-2mm}
\begin{center}
\begin{tabular}{lccccc}
\toprule
Dataset & $n$ & $p$ & GRASP & GIGG & SparseGL \\
\cmidrule(r){1-1}\cmidrule(lr){2-3}\cmidrule(l){4-6}
Covid-19 Trust & 6831 & 101 & \textbf{33.80} (0.09) & 33.81 (0.09) & 34.57 (0.09) \\
Economic Crisis & 50 & 71 & \textbf{33.90} (0.98) & 35.48 (0.99) & 43.99 (1.66) \\
\bottomrule
\end{tabular}   
\end{center}
\vspace{-1mm}
\end{table*}

\subsection{Trust in Experts During COVID-19: A U.S. Survey Dataset} \label{sec: trust_results}
Following the experimental setup outlined in \citet{liang2022sparsegl}, we analyze the data from The Delphi Group at Carnegie Mellon University U.S. COVID-19 Trends and Impact Survey. The response variable, \texttt{trust in experts}, represents the estimated percentage of respondents who trust experts such as doctors, the Centers for Disease Control (CDC), and government health officials. The dataset includes five categorical — month of report, state of residence, age group, race/ethnicity, and self-reported gender identity; along with two continuous variables: 
\begin{itemize}
    \item \texttt{hh\_cmnty\_cli}: The percentage of people who know someone with COVID-like illness in their community (including their household)
    \item \texttt{cli}: The percentage of people experiencing COVID-like illness
\end{itemize}
To capture potential nonlinear effects, we apply a B-spline basis expansion with 10 degrees of freedom (9 knots with cubic splines) to the two continuous variables. 

This experiment highlights two common applications of grouped regression: (1) using discrete factors as predictors, and (2) employing additive models where continuous predictors are expanded into basis functions. As a result, the dataset includes seven predictor groups (5 categorical factors and 2 continuous basis expansions). The design matrix is constructed as described in \citet{liang2022sparsegl} or can be found in the code on our GitHub page \footnote{https://github.com/shuyu-tew/GRASP.git}. After converting the categorical variables to dummy variables and including the B-spline expansions, the final design matrix contains 9759 observations and 101 predictors.

Figure \ref{fig: non-linear effect} illustrates the estimated nonlinear effects of the two continuous predictors. Both GIGG and GRASP produce similar nonlinear trends; however, GRASP exhibits narrower 95\% credible intervals than GIGG, particularly in regions with lower data density, indicating greater confidence in its estimates. Notably, the SparseGL implementation of the sparse group lasso does not inherently provide such credible intervals.

For \texttt{hh\_cmnty\_cli}, the general effect trend starts with a plateau at low values before becoming increasingly positive as the percentage of people who know someone with COVID-like illness in their community increases. This suggests that individuals tend to trust experts more when they perceive a higher prevalence of COVID-like illness around them. However, SparseGL captures a slightly different pattern at lower values: initially increasing for \texttt{hh\_cmnty\_cli} below 5\%, then the estimated effect dipped below zero before following the general upward trend. This initial deviation, however, remains within the 95\% credible intervals of GRASP and GIGG. Notably, we observe a similar dip for GRASP and GIGG, where the effect briefly falls below zero for small values of \texttt{hh\_cmnty\_cli} before rising again for larger values.

As for the effect of \texttt{cli} (percentage of people experiencing COVID-like illness), it is predominantly negative, with the magnitude of the effect increasing for larger \texttt{cli} values. This indicates that personal experiences with COVID-like illness may erode trust in experts.

\begin{figure}[t]
     \centering
     \begin{subfigure}[b]{0.48\textwidth}
         \centering
         \includegraphics[width=\textwidth]{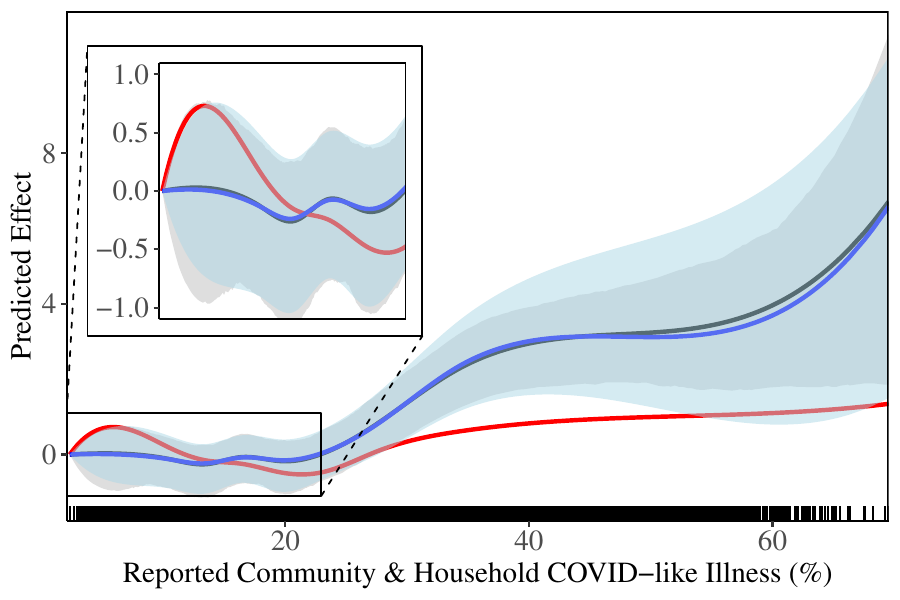}
         \caption{}
     \end{subfigure}
     \hspace{1em}%
     \begin{subfigure}[b]{0.48\textwidth}
         \centering
         \includegraphics[width=\textwidth]{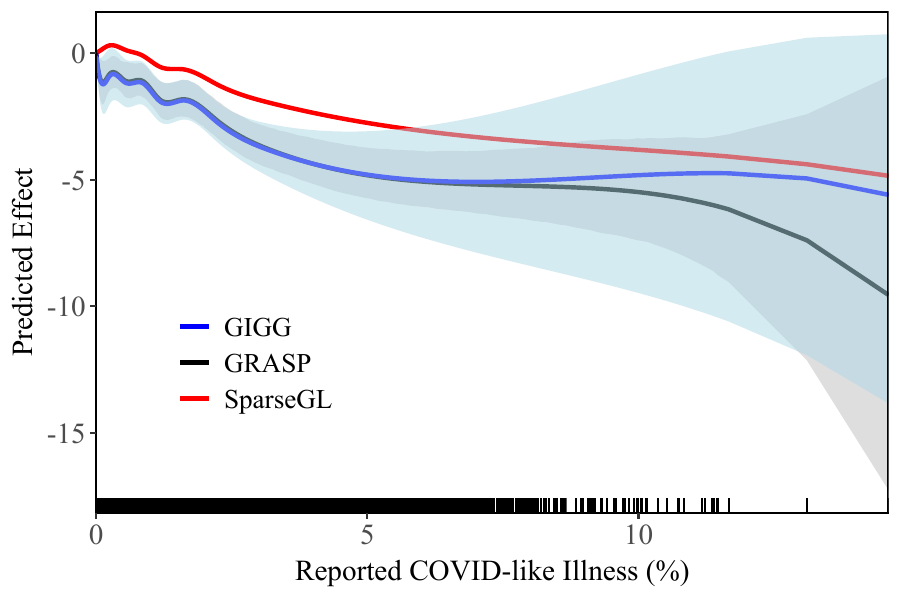}
         \caption{}
     \end{subfigure}
     \caption{Estimated nonlinear effects of (a) \texttt{hh\_cmnty\_cli} - the percentage of people who know someone with COVID-like illness in their community; and (b) \texttt{cli} - the percentage of people experiencing COVID-like illness, conditional on all other variables in the model being set to zero. } \label{fig: non-linear effect}
\end{figure}

\subsection{Economic Dataset: 2008–09 Crisis Cross-Country Analysis} \label{sec: economic_results}
To further evaluate the performance of GRASP, we analyze the balanced economic dataset used in \citet{lee2015bsgs, ho2015Looking}. This dataset is a refined version of the one originally collected by \citet{rose2011cross}, with missing data removed. This study investigates the factors contributing to cross-country variation in the macroeconomic impact of the 2008–09 global financial crisis. The response variable, real GDP growth rate from 2008 to 2009, serves as a measure of economic resilience during the crisis. The dataset includes 51 predictor variables for a sample of 72 countries, categorized into nine theoretical groups representing different potential origins of the crisis. The groups and the number of variables within each group (indicated in parentheses) are as follows: principal factors (10), financial policies (3), financial conditions (4), asset price appreciation (2), macroeconomic policies (4), institutions (11), geography (4), financial linkages (1), and trade linkages (12). A detailed breakdown of these groups and the variables within each group is provided in Table \ref{tab:economic_variables}, while the definitions of each variable can be found in Table 1 of \citet{ho2015Looking}.

Table \ref{tab:economic_variables} compares the average coefficients estimated by GRASP, GIGG, and SparseGL, along with the percentage of times the groups and variables were selected by SparseGL over 100 repetitions of the experiment (70-30 train-test splits). The two most frequently selected groups in SparseGL are ``Institutions'' and ``Trade Linkages'', aligning partially with findings from \citet{lee2015bsgs}, who identify ``Financial Policies'' and ``Trade Linkages'' as key contributors to the crisis. Our results are consistent with those of \citet{ho2015Looking}, who also identify specific variables within the ``Institutions'' group as important. The percentage of times variables are selected by SparseGL generally matches the ranking of variable importance reported in \citet{ho2015Looking}, lending credibility to our findings.

A limitation of GRASP and GIGG is their inability to produce sparse solution. However, this limitation is context-dependent. In our experiments, we observe that GRASP and GIGG achieve significantly lower prediction errors compared to SparseGL. This is expected, as taking the posterior mean typically improves predictive performance by including all variables in the model for prediction, rather than inducing sparsity. While this approach enhances prediction accuracy, it may come at the cost of model simplicity and interpretability, as the resulting models are not sparse.

That said, we note that for groups selected less than 1\% of the time by SparseGL (e.g., ``Asset Price Appreciation'' and ``Financial Conditions''), the corresponding coefficients estimated by GRASP and GIGG are close to zero (i.e., $<|0.01|$ ). Sparsity can be indirectly achieved when GRASP and GIGG are paired with thresholding methods or post hoc Bayesian variable selection techniques, such as the Decoupled shrinkage and selection (DSS) procedure or Stochastic Search Variable Selection (SSVS). These methods leverage the fact that GRASP and GIGG effectively shrink the coefficients of less relevant variables toward zero, enabling variable selection in a subsequent step. Nevertheless, explicit sparsity may still be desirable in certain applications to improve interpretability and practical utility. Future work will explore incorporating sparsity into GRASP using the EM algorithm proposed in \citet{tew2022sparse}. This extension would allow GRASP to balance the trade-off between prediction accuracy and model interpretability, making it more versatile for a wider range of applications.
\begin{tabularx}{\textwidth}{@{}>{\small}l>{\small}c>{\small}c>{\small}c@{}}
\caption{The nine groups of predictors and their corresponding variables used in the analysis of the 2008–09 financial crisis dataset. The values represents the average estimated coefficients from GRASP, GIGG, and SparseGL, along with the percentage of times each group and variable is selected by SparseGL over 100 repetitions of the experiment (70-30 train-test split) in parentheses.} \label{tab:economic_variables} \\
\toprule
Variable & GRASP & GIGG & sparsegl (selected \%) \\
\midrule
\underline{\textit{Principal Factors}} \\
Log (2006 population) Trust & 0.05 & 0.07 & $<|0.01|$ (6)  \\
Log (2006 real GDP per capita) & -0.42 & -0.44 & $<|0.01|$ (6)\\
OECD dummy & -0.41 & -0.22 & -0.04 (6) \\
High-income, non-OECD dummy & 0.04 & 0.05 & $<|0.01|$ (6)\\
Developing East Asia, Pacific dummy & 0.40 & 0.09 & 0.05 (6)\\
Developing Eastern Europe, Central Asia dummy & -0.63 & -0.31 & -0.04 (6)\\
Developing Latin American, Caribbean dummy & -0.05 & -0.03 & 0.01 (5) \\
Developing South Asia dummy & 0.90 & 0.07 & 0.06 (5) \\
Developing Sub-Saharan Africa dummy & 0.78 & 0.10 & 0.05 (4)\\
Stock market growth, 2003-06 & 0.17 & 0.14 & 0.03 (6) \\ \vspace{-2mm}
\\
\underline{\textit{Financial policies}} \\
Credit market regulation, 2006 & -0.73 & -0.54 & $<|0.01|$ (1)\\
Private bank ownership, 2006 & -0.17 & -0.23 & $<|0.01|$ (1)\\
Interest rate controls, 2006 & $<|0.01|$ & $<|0.01|$ & 0 (0) \\ \vspace{-2mm}
\\
\underline{\textit{Financial conditions}} \\
Domestic credit private sector, \% GDP 2006 & $<|0.01|$ &$<|0.01|$ & $<|0.01|$ (1) \\
Domestic bank credit, \% GDP 2006 & $<|0.01|$ & $<|0.01|$ & $<|0.01|$ (1) \\
Private sector credit assess, 2006 & $<|0.01|$ & $<|0.01|$ & $<|0.01|$ (1) \\
Bank liquid reserves, \% assets 2006 & $<|0.01|$ & $<|0.01|$ & $<|0.01|$(1) \\ \vspace{-2mm}
\\
\underline{\textit{Asset price appreciation}} \\
Market capitalization, \% GDP 2006 & $<|0.01|$ & $<|0.01|$ & 0 (0) \\
Stocks traded, \% GDP 2006 & $<|0.01|$ & $<|0.01|$ & 0 (0) \\ \vspace{-2mm}
\\ 
\underline{\textit{Macroeconomic policies}} \\
Currency union member, 2006 & -0.27 & -0.06 & -0.02 (1) \\
EU but non-EMU member, 2006 & -0.04 & -0.01 & -0.01 (1) \\
Inflation targeter, 2006 & -0.02 & $<|0.01|$ & $<|0.01|$(1) \\
CPI inflation, 2006 & -0.07 & -0.08 & $<|0.01|$ (1) \\ \vspace{-2mm}
\\ 
\underline{\textit{Institutions}} \\
Credit/labour/business regulation, EFW 2006 & -0.09 & -0.07 & $<|0.01|$ (48) \\
Overall economic freedom, 2006 & 0.06 & 0.01 & $<|0.01|$ (48)\\
Common law country & 0.01 & 0.03 & -0.01 (41)\\
Control of corruption & 0.14 & 0.13 & -0.05 (48)\\
Regulatory quality & -0.42 & -0.23 & -0.09 (48)\\
Rule of law & 0.13 & 0.13 & -0.06 (48)\\
Political rights, 2006 & 0.18 & 0.20 & 0.02 (48)\\
Civil liberties, 2006 & 0.40 & 0.45 & 0.02 (48)\\
Government size, 2006 & 0.18 & 0.21 & $<|0.01|$ (46)\\
Legal security of property rights, 2006 & -0.23 & -0.26 & $<|0.01|$ (48)\\
Sound money access, 2006 & 0.21 & 0.14 & $<|0.01|$ (46) \\ \vspace{-2mm}
\\ 
\underline{\textit{Geography}} \\
Log of latitude & -0.18 & -0.16 & $<|0.01|$ (13)\\
East Asian dummy & 0.14 & 0.04 & 0.04 (13)\\
Commodity exporter dummy & 0.24 & 0.17 & 0.03 (13)\\
English language dummy & -0.27 & -0.07 & -0.02 (13)\\ \vspace{-2mm}
\\
\underline{\textit{Financial linkages}} \\
BIS consolidated banking share (exposure to UK) & -13.8 & -0.35 & -1.96 (14) \\ \vspace{-2mm}
\\
\underline{\textit{Trade linkages}} \\
Exports exposure to USA & $<|0.01|$ & $<|0.01|$ & $<|0.01|$ (78)\\
Exports exposure to UK & 0.03 & 0.03 & -0.01 (84)\\
Exports exposure to Germany & -0.01 & -0.02 & -0.02 (94)\\
Exports exposure to Japan & 0.08 & 0.06 & 0.06 (94)\\
Exports exposure to Spain & 0.07 & 0.04 & 0.04 (92)\\
Exports exposure to small crises & -0.23 & -0.24 & -0.15 (94)\\
Trade exposure to USA & -0.02 & -0.04 & $<|0.01|$ (85)\\
Trade exposure to UK & -0.09 & -0.08 & -0.05 (94)\\
Trade exposure to Germany & -0.02 & -0.03 & -0.02 (94)\\
Trade exposure to Japan & 0.13 & 0.11 & 0.07 (94)\\
Trade exposure to Spain & -0.02 & -0.02 & 0.02 (89)\\
Trade exposure to small crises & -0.43 & -0.46 & -0.16 (94)\\
\bottomrule
\end{tabularx}

\section{Discussion}\label{sec:conclusion}

In this paper, we introduced a novel adaptive grouped regression method with the normal beta prime (NBP) prior. Our experimental results demonstrate that the proposed method is highly competitive, even in standard linear regression settings where group structure information is absent. Specifically, RASP outperforms existing adaptive normal-beta prime estimators, such as the one implemented in the \textsc{nbp} R package \cite{bai2021beta}, in terms of both predictive performance and computational efficiency. 

In the context of grouped regression, GRASP provides a conceptually straightforward framework while maintaining strong predictive accuracy. Notably, in scenarios with lower sparsity levels among the coefficients and lower signal, our method proves to give better predictive performance be highly competitive with the state-of-the-art GIGG prior \cite{boss2023group}. This suggests that the key to achieving superior performance in grouped regression lies in introducing adaptability into the shrinkage parameters through the use of adaptive priors. We argue that any prior exhibiting behavior similar to the NBP prior could serve this purpose effectively. The GIGG hierarchy proposed by \citet{boss2023group} is an elegant construction in which, when there is only a single group ($G=1$), the marginal prior on regression coefficients reduces to the horseshoe prior. Their work demonstrates the superior performance of GIGG compared to the grouped half-Cauchy hierarchy of \citet{xu2016bayesian}. However, this comparison may not be entirely fair, as \citet{xu2016bayesian} does not incorporate adaptivity in its shrinkage parameters. Our results show that once adaptivity is introduced in both local and group-level shrinkage parameters — by assigning the beta prime prior to both — the conceptually simpler hierarchy of \citet{xu2016bayesian} can perform as effectively as GIGG. Additionally, the proposed grouped NBP estimator has the potential to handle multiple and overlapping group structures, a setting for which it is not immediately clear how the GIGG hierarchy could be extended.

Another key contribution of this work is the introduction of a framework to explicitly quantify the correlation between effective local shrinkage parameters within a group. This framework allows us to explore the behavior of these correlations under different prior configurations, providing new insights into the interplay between group structure and shrinkage behavior. To the best of our knowledge, this is the first time such a framework has been proposed, and it opens up new avenues for understanding and designing grouped regression models. Future work will explore extending our method to incorporate sparsity directly into the grouped NBP framework, potentially enhancing its interpretability while preserving its predictive power.


\section{Acknowledgments and Disclosure of Funding}
This work was supported by the Australian Research Council (DP210100045).

\newpage

\appendix

\section{Proof of Proposition \eqref{propo:corr_delta_lambda}} \label{apx: prop_proof}
The correlation between two random variables $X$ and $Y$ is defined as:
\begin{equation}
    {\rm corr}(X, Y) = \frac{\E{}{XY} - \E{}{X}\E{}{Y}}{ \sqrt{\V{}{X} \V{}{Y}} }
\end{equation}
To quantify the correlation between the effective shrinkage parameters within a group, we are interested in ${\rm corr}(\log(\delta_g\lambda_{gi}), \log(\delta_g\lambda_{gj}))$. Taking the logarithm of the product of two variables is equivalent to adding their logarithms. Therefore, we can write $\log(\delta_g\lambda_{gj}) = \log(\delta_g) + \log(\lambda_{gj})$. Let $d \equiv \log \delta$, $l_1 \equiv \log \lambda_{g1}$ and $l_2 \equiv \log \lambda_{g2}$, the correlation between the shrinkage factors within the same group becomes
\begin{align*}
    {\rm corr}(l_1 + d, l_2 + d) &= \frac{\E{}{(l_1 + d)(l_2 + d)} - \E{}{(l_1 + d)}\E{}{(l_2 + d)}}{ \sqrt{\V{}{l_1 + d} \V{}{l_2 + d}} }\\
    &= \frac{\V{}{d}}{ \sqrt{(\V{}{d} + \V{}{l_1}) (\V{}{d} + \V{}{l_2})} }
\end{align*}
Since $l_1$ and $l_2$ belong to the same group, they share the same hyperparameters, hence $\V{}{l_1} = \V{}{l_2} = \V{}{\log(\lambda_g)}$. This simplifies the correlation expression to:
\begin{equation} \label{eq_apx: corr}
    {\rm corr}(l_1 + d, l_2 + d) = \frac{\V{}{d}}{ \V{}{d} + \V{}{\log(\lambda_g)} }
\end{equation}

\subsection{Application to GIGG} \label{apx: prop_GIGG}
The group inverse-gamma gamma (GIGG) hierarchy proposed by \cite{boss2023group} is defined as follows
\begin{align*}
\begin{split}
  \lambda^2_{gj} \; &\sim \; {\rm IG}(b_g, 1) \\
  \delta^2_g \; &\sim \; {\rm Gamma}(a_g, 1) \\
\end{split}
\end{align*}
To compute the correlation between the effective log transformed shrinkage parameters within a group (i.e. ${\rm corr}(\log(\delta^2_g\lambda^2_{gi}), \log(\delta^2_g\lambda^2_{gj}))$), we use the following properties:
\begin{itemize}
    \item If $X \sim {\rm Gamma}(a, \theta)$, then $\V{}{\log X} = \phi'(a)$.
    \item For $X \sim {\rm IG}(a, 1/\theta)$, it follows that $1/X \sim {\rm Gamma}(a, \theta)$. The variance of $\log X$ is equivalent to that of its reciprocal since $\log(1/X) = - \log X$, and $\V{}{- \log X} = \V{}{\log X} = \phi'(a)$.
\end{itemize}
Applying these properties, the variances for the log-transformed local and group shrinkage parameters under the GIGG hierarchy are:
\begin{align*}
    \V{}{\log \lambda^2_{gj}} &= \phi'(a_g)\\
    \V{}{\log \delta^2_{g}} &= \phi'(b_g).
\end{align*}
Substituting these into the derived correlation formula in Equation \ref{eq_apx: corr}, we obtain the exact correlation under the GIGG hierarchy:
\begin{equation*}
    {\rm corr}(\log(\delta^2_g\lambda^2_{gi}), \log(\delta^2_g\lambda^2_{gj})) = \frac{\phi'(a_g)}{\phi'(a_g) + \phi'(b_g)}.
\end{equation*}

\subsection{Application to GRASP} \label{apx: prop_GNBP}
The proposed GRASP hierarchy consist of placing the beta prime prior at both the local and group shrinkage parameter as given in Equation \ref{hier: betaP_prior} (repeated here for convenience):
\begin{align*}
\begin{split}
  \lambda^2_{gj} \; &\sim \; B'(a_g,b_g) \\
  \delta^2_g \; &\sim \; B'(a, b) \\
\end{split}
\end{align*}
To compute the correlation between the effective log transformed shrinkage parameters within a group (i.e. ${\rm corr}(\log(\delta^2_g\lambda^2_{gi}), \log(\delta^2_g\lambda^2_{gj}))$), we utilize the fact that if $X \sim B'(a, b)$, then $\log X$ follows a generalized logistic distribution. The variance of this distribution is then given by $\V{}{\log X} = \phi'(a) + \phi'(b)$. From this property, the variances for the log-transformed local and group shrinkage parameters can be expressed as:
\begin{align*}
    \V{}{\log \lambda^2_{gj}} &= \phi'(a_g) + \phi'(b_g)\\
    \V{}{\log \delta^2_{g}} &= \phi'(a) + \phi'(b).
\end{align*}
Substituting these into the derived correlation formula in Equation \ref{eq_apx: corr}, we obtain the exact correlation under the GRASP hierarchy:
\begin{equation*}
    {\rm corr}(\log(\delta^2_g\lambda^2_{gi}), \log(\delta^2_g\lambda^2_{gj})) = \frac{\phi'(a) + \phi'(b)}{\phi'(a) + \phi'(b) + \phi'(a_g) + \phi'(b_g)}.
\end{equation*}

\section{Fast and Accurate Approximation of the Beta Shape
Parameters} \label{apx: Fast and Accurate Approximation of the Beta Shape Parameters}
Given $x_1, \cdots, x_n | a, b \sim {\rm Beta}(a, b)$ and a prior for the shape parameters $a$ and $b$, the goal is to approximate the conditional distribution of the shape parameters $a$ and $b$ using a Gamma distribution (i.e. $p(a | x_1, \cdots, x_n, b) \approx {\rm Gamma(k, \theta)}$ and $p(b | x_1, \cdots, x_n, a) \approx {\rm Gamma(k, \theta)}$).

More generally, this proposed algorithm can approximate any target distribution $f(y)$ with an unknown normalizing constant using a proposal distribution from the exponential family. The main idea is to estimate the natural parameters of the proposal distribution that best approximate the target distribution (minimizes the difference between the proposal and target distributions).

\subsection{Algorithm}
To approximate or sample from a target probability distribution $p(x)$, this section presents an algorithm that utilizes a proposal distribution chosen from the exponential family. The goal is to estimate the parameters of this proposal distribution such that it is as close as possible to the target distribution. Consider a parameter vector $\btheta = (\theta_1, \cdots, \theta_m)$. A probability distribution belongs to the exponential family if its probability density function can be expressed as
\begin{equation} \label{eq: exponential_family}
  \begin{aligned}
    q(x | \btheta) = h(x) \exp\left( \bEta(\btheta) \trans \bT(x)  - A(\btheta)\right)
  \end{aligned}
\end{equation}
where $\bT(x)$ is the sufficient statistics of the proposal distribution, $A(\btheta)$ is the normalizing constant that integrates \eqref{eq: exponential_family} to 1 over the sample space of $x$, and $h(x)$ can be viewed as a constant with respect to $\btheta$. The term $\bEta = \bEta(\btheta)$ is defined as the natural parameter, a transformed parameter that converts the exponential family distributions to its canonical form, allowing it to be expressed as a linear function of its sufficient statistics. For example, in the gamma distribution, which is typically parameterized by the shape $k$ and scale $\omega$, the natural parameters are $\bEta(k, \omega) = (k-1, -1/\omega)$, and the sufficient statistics are $(\sum_{i=1}^n \log(x), \sum_{i=1}^n x)$.

The log of equation \eqref{eq: exponential_family} can then be simplified to:
\begin{equation}\label{eq: log_exponential_family}
  \begin{aligned}
    \log q(x | \btheta) &= \log h(x) + \bEta\trans\bT(x)  - A(\btheta) \\
    & = \bEta\trans\bT(x) + c
  \end{aligned}
\end{equation}
where $c$ is the linear combination of the log constants $\log h(x) - A(\btheta)$. Notable, Equation \eqref{eq: log_exponential_family} takes the form of a linear regression. By equating $\log p(x)$ (the log of the target distribution) to the RHS of Equation \eqref{eq: log_exponential_family}:
\begin{equation}\label{eq: linear_log_exponential_family}
    \log p(x) = \bEta\trans\bT(x) + c,
\end{equation}
and solving for the natural parameters $\bEta$ using least squares, we can find the parameters of the proposal distribution that best approximate the target distribution.


The algorithm is summarized in the following steps:
\begin{enumerate}
    \item Select a proposal distribution from the exponential family and express it in the form of Equation \eqref{eq: log_exponential_family} with natural parameters
    \item Choose at least three design points, $x_1, \cdots, x_m$ from the target distribution. Refer to Section \ref{sec: choosing design points} for suggestions on choosing these design points.
    \begin{enumerate}
        \item Compute $\log p(x)$ at each of these points:
        \begin{equation}
            \by = \left[\log p(x_1), \cdots, \log p(x_m) \right]\trans
        \end{equation}
        \item Compute the sufficient statistics $\bT(x)$:
        \begin{equation}
            \bX = \left[\bT(x_1), \cdots, \bT(x_m) \right]\trans
        \end{equation}
    \end{enumerate}

    \item Substitute these values into Equation \eqref{eq: linear_log_exponential_family} to get:
        \[
        \by = \bX\bEta + c.
        \]
        Solve this linear system using least squares to estimate the $\bEta$. If $\bX$ is centered, the intercept term $c$ can be omitted and the least squares estimate is $\hat{\bEta}_{\rm LS} = (\bX\trans\bX)\inv(\bX\trans\by)$. For distributions with more than two sufficient statistics, we suggest selecting at least one additional design point for each additional sufficient statistic.
\end{enumerate}
If the proposal distribution exactly matches the target distribution, this method provides an exact fit. Otherwise, it yields an approximation that minimizes the difference between the proposal and target distributions, typically measured by Kullback-Leibler divergence or mean squared error.

The algorithm described above returns the parameters of the proposal distribution, $\btheta^*$ that best approximate the target distribution. Given these parameters, random samples from the approximated target distribution can be generated using the Metropolis-Hastings algorithm. This involves updating the current value
$x$ by sampling a new proposal $x' \sim q(\btheta^*)$ and accepting the proposal with probability:
\begin{equation}\label{eq:  MH_prob}
    {\rm min}\left\{ 1, \frac{p(x')q(x|\btheta^*)}{p(x)q(x'|\btheta^*)} \right\}.
\end{equation}

This method can be seamlessly integrated into Gibbs sampling procedures. For instance, it can be used to estimate the shape parameters of the beta distribution in the following hierarchical model:
\begin{equation}
    \begin{aligned}
    x_1, \cdots, x_n | a, b &\sim {\rm Beta}(a, b) \\
    a &\sim {\rm C}^+(0,1) \\
    b &\sim {\rm C}^+(0,1) 
  \end{aligned}
\end{equation}
where $C^+(0,1)$ denotes the half-Cauchy distribution. The standard Gibbs sampling approach involves repeatedly iterating:
\begin{enumerate}
    \item Sample $a$ from $p(a| x_1, \cdots, x_n, b)$. \\
    \\
    The target distribution here is the posterior conditional density of $a$, expressed as:
    \begin{equation} \label{eq: log_posterior_a}
        \begin{aligned}
            \log p(a|x_1, \cdots, x_n, b) = (a-1) \sum^n_{i=1} \log (x_i) - n \log {\rm B}(a,b) - \log(a^2 + 1)
        \end{aligned}
    \end{equation}
    where ${\rm B}(\cdot)$ is the beta function. We approximate this distribution using a Gamma distribution (i.e. $p(a | x_1, \cdots, x_n, b) \approx {\rm Gamma(k, \theta)}$). Using the algorithm described earlier (or specifically detailed in Algorithm \ref{algo: beta_approx}), we estimate $k$ and $\theta$, sample a proposal $a' \sim {\rm Gamma}(k, \theta)$ then accept with a probability defined as before.

    \item Sample $b$ from $p(b| x_1, \cdots, x_n, a)$ \\
    \\
    This step is similar to the first, but here the target distribution is the log of the posterior conditional density of $b$:
    \begin{equation} \label{eq: log_posterior_b}
        \begin{aligned}
            \log p(b|x_1, \cdots, x_n, a) = (b-1) \sum^n_{i=1} \log (1-x_i) - n \log {\rm B}(a,b) - \log(b^2 + 1)
        \end{aligned}
    \end{equation}
    The key difference between the Equation \eqref{eq: log_posterior_a} and \eqref{eq: log_posterior_b} lies in the first term. Notably, if $x \sim {\rm Beta} (a,b)$, then $1-x \sim {\rm Beta} (b,a)$. Therefore, the same computational routine developed for sampling $a$ can be applied to $b$, with adjustments to the input parameters, from $a$ to $b$, and from $\sum^n_{i=1} \log (x_i)$ to $\sum^n_{i=1} \log (1-x_i)$, as outlined in Algorithm \ref{algo: beta_sampler}.
\end{enumerate}


\subsection{Choosing the design (data) points of the target distribution} \label{sec: choosing design points}
When selecting the design points, it is important that the points are evenly distributed across the target distribution. This ensures that the approximation captures the correct overall shape of the target distribution, even if the proposal distribution differs greatly from the target. If the design points are clustered in one region, a poor choice of proposal distribution could result in worse approximations. While selecting more design points generally improves the fit, it can reduce the efficiency of the sampler, especially when computational resources are limited.

One effective strategy is to include the mode of the target density as one of the design points. For example, in the case of estimating the parameter $a$ from the conditional posterior of the following hierarchical model:
\begin{equation}
    \begin{aligned}
    x_1, \cdots, x_n | a, b &\sim {\rm Beta}(a, b) \\
    a &\sim {\rm C}^+(0,1)
  \end{aligned}
\end{equation}
where $C^+(0,1)$ denotes the half-Cauchy distribution, the mode of $a$ will be close to the maximum likelihood estimate. Therefore, we can treat the conditional posterior as a likelihood and estimate the mode of $a$ using maximum likelihood estimation. Since there is no closed-form solution for the derivative of the negative log-likelihood, we solve it iteratively using the Newton-Raphson method. The negative log of the posterior conditional density is:
\begin{equation}
    \begin{aligned}
    - \log p(a|x_1, \cdots, x_n, b) = (1-a) \sum^n_{i=1} \log (x_i) + n \log {\rm B}(a,b) + \log(a^2 + 1)
  \end{aligned}
\end{equation}
where ${\rm B}(\cdot)$ is the beta function. The gradient and the Hessian of this target density is then
\begin{equation}
    \begin{aligned}
    \label{eq: gradient_hessian}
    \frac{\partial}{\partial a} - \log p(a|x_1, \cdots, x_n, b) &=  \sum^n_{i=1} \log (x_i) - n \phi(a) - n\phi(a+b) - \frac{2a}{1+a^2},\\
    \frac{\partial^2}{\partial a^2} - \log p(a|x_1, \cdots, x_n, b) &=  - n \phi'(a) - n\phi'(a+b) - \frac{2}{1+a^2} + \frac{4a^2}{(1+a^2)^2},\\
  \end{aligned}
\end{equation}
where $\phi(\cdot)$ and $\phi'(\cdot)$ is the digamma and trigamma function respectively. We then apply the Newton-Raphson method as outlined in Algorithm \ref{algo: Newton-Raphson}, with updates derived in Equation \eqref{eq: gradient_hessian} included.

Once the mode $a_0$ is identified, two additional design points are selected to be approximately one standard deviation away (i.e. $\ba = \left({\rm max}(a_0 - \sqrt{2/H}, a_0/2), a_0, a_0 + \sqrt{2/H} \right)$, where $H$ is the Hessian of the target distribution). For more design points, values two or three standard deviations away can be used, with additional points subsampled between the mode and these extended points.

\begin{algorithm}
    \DontPrintSemicolon
    \SetKwInOut{KwIn}{Input}
    \SetKwInOut{KwOut}{Output}
    \SetKwComment{Comment}{}{}
    
    \KwIn{Samples from target distribution $x_1, \cdots, x_n > 0$, shape parameter $b$}
    \KwOut{Three design points of the target distribution ($a_1, a_2, a_3$). $a_2$ is the estimated mode.}

    \vspace*{3mm}

    \SetKwFunction{FMain}{a\_est}
    \SetKwProg{Fn}{Function}{:}{}
    \Fn{\FMain{$x_1, \cdots, x_n, b$}}{

    ${\rm slogx} = \sum^n_{i = 1}\log(x_i + {\rm eps})$ {\scriptsize \tcp*[f]{eps: machine epsilon (in double precision) to handle $x_i$ being 0 or 1} }\\
    ${\rm slogcx} = \sum^n_{i = 1}\log(1 - x_i + {\rm eps})$ \\
    $c = 0.9$  \\
    $\kappa = 1$ {\scriptsize \tcp*[h]{learning rate}} \\
    $\delta = 10^{10}$ \\

    \vspace*{3mm}
    \tcp*[h]{Initial guess} \\
    $a_0 = {\rm max}\left( \frac{1}{2} + \frac{{\rm slogx}/n}{2(1-{\rm slogx}/n - {\rm sclogx}/n)} , 10^{-4}\right)$ \\

    \vspace*{4mm}
    \tcp*[h]{Newton-Raphson update} \\

    \While($\quad \quad \quad \quad$  {\scriptsize \tcp*[h]{we set $\epsilon = 0.001$}}){$\delta < \epsilon$}{

        $\displaystyle g = n(\phi(a_0) - \phi(a_0+b)) - {\rm slogx} + \frac{2a_0}{(1+a_0^2)}$ {\scriptsize \tcp*[h]{Gradient of the target density}}

        $\displaystyle H = n(\phi'(a_0) - \phi'(a_0+b)) + \frac{2}{(1+a_0^2)} - \frac{4a_0^2}{(1+a_0^2)^2}$ {\scriptsize \tcp*[h]{Hessian of the target density}}

        $ a_{\rm new} = a_0 - \kappa g/H$

        \vspace*{3mm}
        \If{$a_{\rm new} < 0$}{
            $\kappa = c\kappa$
        }

        \vspace*{3mm}

        $\delta = |a - a_{\rm new}|/ |a|$

        $a_0 = a_{\rm new}$
    }
    \vspace*{1mm}

    $\displaystyle H = n(\phi'(a_0) - \phi'(a_0+b)) + \frac{2}{(1+a_0^2)} - \frac{4a_0^2}{(1+a_0^2)^2}$ {\scriptsize \tcp*[h]{Hessian at the maximum}}

    $(a_1, a_2, a_3) = \left({\rm max}(a_0 - \sqrt{2/H}, a_0/2), a_0, a_0 + \sqrt{2/H} \right)$
    
    \vspace*{3mm}
    \KwRet{($a_1, a_2, a_3$)}
    }
    \textbf{end Function}
    \caption{Newton–Raphson procedure for estimating the mode and design points.}\label{algo: Newton-Raphson}
\end{algorithm}

\begin{algorithm}
    \DontPrintSemicolon
    \SetKwInOut{KwIn}{Input}
    \SetKwInOut{KwOut}{Output}
    \SetKwComment{Comment}{}{}
    {\footnotesize \tcc{This funtion uses Algorithm \ref{algo: Newton-Raphson} to approximate the distribution's parameters.}}

    \vspace*{3mm}
    \KwIn{Samples from target distribution $x_1, \cdots, x_n > 0$, previous sample $a_{\rm prev}$, shape parameter $b$}
    \KwOut{A single proposed sample from the target distribution, $a_{\rm new}$. }
    
    \vspace*{3mm}
    \SetKwFunction{FMain}{beta\_approx}
    \SetKwProg{Fn}{Function}{:}{}
    \Fn{\FMain{$x_1, \cdots, x_n, a_{\rm prev}, b$}}{
    \vspace*{1mm}
    $(a_1, a_2, a_3)$ = a\_est($x_1, \cdots, x_n, b$) {\scriptsize \tcp*[h]{Get the MLE as starting gamma approximation}}\\
    \vspace*{3mm}
    $\bf X$ = [log($a_1, a_2, a_3$), ($a_1, a_2, a_3$)] {\scriptsize \tcp*[h]{sufficient stats}}\\
    $y_j =  (a_j-1) \sum^n_{i=1} \log (x_i) - n \log {\rm B}(a_j,b) - \log(a_j^2 + 1)$ {\scriptsize \tcp*[h]{log target}}\\
    \vspace*{3mm}

    $\bX$ = $\bX$ - ColumnMeans($\bX$) {\scriptsize \tcp*[h]{Centers the columns of matrix $\bX$}}.\\
    $(\eta_1, \eta_2) = (\bX^{\rm T}\bX)^{-1}\bX^{\rm T} \by$ \\
    $(k, \omega) = (\eta_1 + 1, -1/\eta_2)$ {\scriptsize \tcp*[h]{transform the natural parameter back to shape and scale}}\\

    \vspace*{4mm}
    Sample $a' \sim {\rm Ga}(k, \omega)$; \\
    Compute the acceptance probability 
    \begin{equation*}
        r = {\rm min}\left\{ 1, \frac{p(x')q(x|\btheta^*)}{p(x)q(x'|\btheta^*)} \right\},
    \end{equation*}
    \\
    Sample $u\sim U(0,1)$; \\
    Set new sample to 
    \begin{equation*}
        a_{\rm new} = \begin{cases}
                a' & \text{u<r}\\
                a_{\rm prev} & \text{otherwise}
                    \end{cases}       
    \end{equation*}
    \\
    \vspace*{3mm}
    \KwRet{$a_{\rm new}$}
    }
    \textbf{end Function}
    \caption{Estimation of Beta distribution parameters.}\label{algo: beta_approx}  
\end{algorithm}

\begin{algorithm}
    \DontPrintSemicolon
    \SetKwInOut{KwIn}{Input}
    \SetKwInOut{KwOut}{Output}
    \SetKwComment{Comment}{}{}

    \KwIn{Samples from the target distribution $x_1, \cdots, x_n > 0$, number of desired samples, $n$}
    \KwOut{Drawn posterior conditional samples of the shapes parameters (${\bf a} = a^{(0)}, \cdots, a^{(n)}$) and (${\bf b} = b^{(0)}, \cdots, b^{(n)}$) }

    \vspace*{3mm}
    $a = 1$ \\
    $b = 1$ \\

    ${\rm slogx} = \sum^n_{i = 1}\log(x_i + {\rm eps})$ {\scriptsize \tcp*[f]{eps: machine epsilon (in double precision) to handle $x_i$ being 0 or 1} }\\
    ${\rm slogcx} = \sum^n_{i = 1}\log(1 - x_i + {\rm eps})$ \\

    \vspace*{3mm}
    {\footnotesize \tcc{Initial coordinate-wise maximum likelihood search. Initialize $a$ and $b$ with the approximated mode. Increase the number of iterations if $a$ and $b$ is expected to be very big (i.e. in the magnitude of $10^5$)}}
    \For{$i \in \{1, \cdots ,5\}$}{
        $a^{(0)}$ = a\_est($x_1, \cdots, x_n$, b)[2] {\scriptsize \tcp*[f]{mode is the item at the second index} }\\
        $b^{(0)}$ = a\_est($x_1, \cdots, x_n$, a)[2]
    }

    \vspace*{3mm}
    {\footnotesize\tcp*[h]{Gibbs sampling step}}\\
    \For{$i \in \{1, 2, 3 \cdots\}$}{

        Draw sample $a^{(i+1)}$ using beta\_approx(slogx, slogcx, a, b) \\
        Draw sample $b^{(i+1)}$ using beta\_approx(slogx, slogcx, b, a)

    }

    \vspace*{3mm}
    
    \KwRet{$\left({\bf a} = a^{(0)}, \cdots, a^{(n)}\right)$ {\rm and} $\left({\bf b} = b^{(0)}, \cdots, b^{(n)}\right)$}
    \caption{Gibbs Sampler for approximating the conditional distribution of the Beta distribution parameters}\label{algo: beta_sampler}
\end{algorithm}

\newpage
\bibliographystyle{plainnat}
\bibliography{sample}

\begin{thebibliography}{31}
\providecommand{\natexlab}[1]{#1}
\providecommand{\url}[1]{\texttt{#1}}
\expandafter\ifx\csname urlstyle\endcsname\relax
  \providecommand{\doi}[1]{doi: #1}\else
  \providecommand{\doi}{doi: \begingroup \urlstyle{rm}\Url}\fi

\bibitem[Armagan et~al.(2011)Armagan, Clyde, and
  Dunson]{armagan2011generalized}
Artin Armagan, Merlise Clyde, and David Dunson.
\newblock Generalized beta mixtures of gaussians.
\newblock \emph{Advances in neural information processing systems}, 24, 2011.

\bibitem[Armagan et~al.(2013)Armagan, Dunson, and Lee]{armagan2013generalized}
Artin Armagan, David~B Dunson, and Jaeyong Lee.
\newblock Generalized double pareto shrinkage.
\newblock \emph{Statistica Sinica}, 23\penalty0 (1):\penalty0 119, 2013.

\bibitem[Bai and Ghosh(2017)]{bai2017inverse}
Ray Bai and Malay Ghosh.
\newblock The inverse gamma-gamma prior for optimal posterior contraction and
  multiple hypothesis testing.
\newblock \emph{arXiv preprint arXiv:1710.04369}, 2017.

\bibitem[Bai and Ghosh(2019)]{bai2019large}
Ray Bai and Malay Ghosh.
\newblock Large-scale multiple hypothesis testing with the normal-beta prime
  prior.
\newblock \emph{Statistics}, 53\penalty0 (6):\penalty0 1210--1233, 2019.

\bibitem[Bai and Ghosh(2021)]{bai2021beta}
Ray Bai and Malay Ghosh.
\newblock On the beta prime prior for scale parameters in high-dimensional
  bayesian regression models.
\newblock \emph{Statistica Sinica}, 31\penalty0 (2):\penalty0 843--865, 2021.

\bibitem[Bhadra et~al.(2017)Bhadra, Datta, Polson, Willard,
  et~al.]{bhadra2017horseshoe+}
Anindya Bhadra, Jyotishka Datta, Nicholas~G Polson, Brandon Willard, et~al.
\newblock The horseshoe+ estimator of ultra-sparse signals.
\newblock \emph{Bayesian Analysis}, 12\penalty0 (4):\penalty0 1105--1131, 2017.

\bibitem[Bhattacharya et~al.(2015)Bhattacharya, Pati, Pillai, and
  Dunson]{bhattacharya2015dirichlet}
Anirban Bhattacharya, Debdeep Pati, Natesh~S Pillai, and David~B Dunson.
\newblock Dirichlet--laplace priors for optimal shrinkage.
\newblock \emph{Journal of the American Statistical Association}, 110\penalty0
  (512):\penalty0 1479--1490, 2015.

\bibitem[Boss et~al.(2023)Boss, Datta, Wang, Park, Kang, and
  Mukherjee]{boss2023group}
Jonathan Boss, Jyotishka Datta, Xin Wang, Sung~Kyun Park, Jian Kang, and
  Bhramar Mukherjee.
\newblock Group inverse-gamma gamma shrinkage for sparse linear models with
  block-correlated regressors.
\newblock \emph{Bayesian Analysis}, 1\penalty0 (1):\penalty0 1--30, 2023.

\bibitem[Carvalho et~al.(2010)Carvalho, Polson, and
  Scott]{carvalho2010horseshoe}
Carlos~M Carvalho, Nicholas~G Polson, and James~G Scott.
\newblock The horseshoe estimator for sparse signals.
\newblock \emph{Biometrika}, 97\penalty0 (2):\penalty0 465--480, 2010.

\bibitem[Castillo and Mismer(2018)]{castillo2018empirical}
Isma{\"e}l Castillo and Romain Mismer.
\newblock Empirical bayes analysis of spike and slab posterior distributions.
\newblock 2018.

\bibitem[Gelman(2006)]{gelman2006prior}
Andrew Gelman.
\newblock Prior distributions for variance parameters in hierarchical models
  (comment on article by browne and draper).
\newblock \emph{Bayesian analysis}, 1\penalty0 (3):\penalty0 515--534, 2006.

\bibitem[Ho(2015)]{ho2015Looking}
Tai-kuang Ho.
\newblock Looking for a needle in a haystack: Revisiting the cross-country
  causes of the 2008–9 crisis by bayesian model averaging.
\newblock \emph{Economica}, 82\penalty0 (328):\penalty0 813--840, 2015.
\newblock \doi{https://doi.org/10.1111/ecca.12125}.
\newblock URL \url{https://onlinelibrary.wiley.com/doi/abs/10.1111/ecca.12125}.

\bibitem[Lee and Chen(2015)]{lee2015bsgs}
Kuo-Jung Lee and Ray-Bing Chen.
\newblock Bsgs: Bayesian sparse group selection.
\newblock \emph{R Journal}, 7\penalty0 (2), 2015.

\bibitem[Liang et~al.(2024)Liang, Cohen, Sólon~Heinsfeld, Pestilli, and
  McDonald]{liang2022sparsegl}
Xiaoxuan Liang, Aaron Cohen, Anibal Sólon~Heinsfeld, Franco Pestilli, and
  Daniel~J. McDonald.
\newblock sparsegl: An r package for estimating sparse group lasso.
\newblock \emph{Journal of Statistical Software}, 110\penalty0 (6):\penalty0
  1–23, 2024.
\newblock \doi{10.18637/jss.v110.i06}.
\newblock URL
  \url{https://www.jstatsoft.org/index.php/jss/article/view/v110i06}.

\bibitem[Makalic and Schmidt(2015)]{makalic2015simple}
Enes Makalic and Daniel~F Schmidt.
\newblock A simple sampler for the horseshoe estimator.
\newblock \emph{IEEE Signal Processing Letters}, 23\penalty0 (1):\penalty0
  179--182, 2015.

\bibitem[Mallick and Yi(2017)]{mallick2017bayesian}
Himel Mallick and Nengjun Yi.
\newblock Bayesian group bridge for bi-level variable selection.
\newblock \emph{Computational statistics \& data analysis}, 110:\penalty0
  115--133, 2017.

\bibitem[Maruyama and Matsuda(2024)]{maruyama2024minimaxity}
Yuzo Maruyama and Takeru Matsuda.
\newblock Minimaxity under the half-cauchy prior.
\newblock \emph{arXiv preprint arXiv:2406.08892}, 2024.

\bibitem[Polson et~al.(2012)Polson, Scott, et~al.]{polson2012half}
Nicholas~G Polson, James~G Scott, et~al.
\newblock On the half-cauchy prior for a global scale parameter.
\newblock \emph{Bayesian Analysis}, 7\penalty0 (4):\penalty0 887--902, 2012.

\bibitem[Ro{\v{c}}kov{\'a}(2018)]{rovckova2018bayesian}
Veronika Ro{\v{c}}kov{\'a}.
\newblock Bayesian estimation of sparse signals with a continuous
  spike-and-slab prior.
\newblock 2018.

\bibitem[Rose and Spiegel(2011)]{rose2011cross}
Andrew~K Rose and Mark~M Spiegel.
\newblock Cross-country causes and consequences of the crisis: An update.
\newblock \emph{European economic review}, 55\penalty0 (3):\penalty0 309--324,
  2011.

\bibitem[Salimans and Knowles(2013)]{salimans2013fixed}
Tim Salimans and David~A Knowles.
\newblock Fixed-form variational posterior approximation through stochastic
  linear regression.
\newblock 2013.

\bibitem[Schmidt and Makalic(2018)]{schmidt2018log}
Daniel~F Schmidt and Enes Makalic.
\newblock Log-scale shrinkage priors and adaptive bayesian global-local
  shrinkage estimation.
\newblock \emph{arXiv preprint arXiv:1801.02321}, 2018.

\bibitem[Schmidt and Makalic(2019)]{schmidt2019bayesian}
Daniel~F Schmidt and Enes Makalic.
\newblock Bayesian generalized horseshoe estimation of generalized linear
  models.
\newblock In \emph{Joint European Conference on Machine Learning and Knowledge
  Discovery in Databases}, pages 598--613. Springer, 2019.

\bibitem[Simon et~al.(2013)Simon, Friedman, Hastie, and
  Tibshirani]{simon2013sparse}
Noah Simon, Jerome Friedman, Trevor Hastie, and Robert Tibshirani.
\newblock A sparse-group lasso.
\newblock \emph{Journal of computational and graphical statistics}, 22\penalty0
  (2):\penalty0 231--245, 2013.

\bibitem[Song and Liang(2023)]{song2023nearly}
Qifan Song and Faming Liang.
\newblock Nearly optimal bayesian shrinkage for high-dimensional regression.
\newblock \emph{Science China Mathematics}, 66\penalty0 (2):\penalty0 409--442,
  2023.

\bibitem[Tew et~al.(2022)Tew, Schmidt, and Makalic]{tew2022sparse}
Shu~Yu Tew, Daniel~F Schmidt, and Enes Makalic.
\newblock Sparse horseshoe estimation via expectation-maximisation.
\newblock In \emph{Joint European Conference on Machine Learning and Knowledge
  Discovery in Databases}, pages 123--139. Springer, 2022.

\bibitem[van~der Pas et~al.(2017)van~der Pas, Szab{\'o}, and van~der
  Vaart]{van2017adaptive}
St{\'e}phanie van~der Pas, Botond Szab{\'o}, and Aad van~der Vaart.
\newblock Adaptive posterior contraction rates for the horseshoe.
\newblock \emph{Electronic Journal of Statistics}, 11\penalty0 (2):\penalty0
  3196--3225, 2017.

\bibitem[Van Der~Pas et~al.(2014)Van Der~Pas, Kleijn, and Van
  Der~Vaart]{van2014horseshoe}
St{\'e}phanie~L Van Der~Pas, Bas~JK Kleijn, and Aad~W Van Der~Vaart.
\newblock The horseshoe estimator: Posterior concentration around nearly black
  vectors.
\newblock 2014.

\bibitem[Xu et~al.(2016)Xu, Schmidt, Makalic, Qian, and Hopper]{xu2016bayesian}
Zemei Xu, Daniel~F Schmidt, Enes Makalic, Guoqi Qian, and John~L Hopper.
\newblock Bayesian grouped horseshoe regression with application to additive
  models.
\newblock In \emph{Australasian Joint Conference on Artificial Intelligence},
  pages 229--240. Springer, 2016.

\bibitem[Yu et~al.(2021)Yu, Xu, and Cao]{yu2021adaptive}
Hanjun Yu, Xinyi Xu, and Di~Cao.
\newblock The adaptive normal-hypergeometric-inverted-beta priors for sparse
  signals.
\newblock \emph{Journal of Statistical Computation and Simulation}, 91\penalty0
  (2):\penalty0 396--419, 2021.

\bibitem[Yuan and Lin(2006)]{yuan2006model}
Ming Yuan and Yi~Lin.
\newblock Model selection and estimation in regression with grouped variables.
\newblock \emph{Journal of the Royal Statistical Society Series B: Statistical
  Methodology}, 68\penalty0 (1):\penalty0 49--67, 2006.

\end{thebibliography}

\end{document}